\documentclass[a4paper,11pt,]{article}
\usepackage{graphicx}
\usepackage{tabularx}
 \usepackage{amssymb}
 \usepackage{amsmath}
\usepackage{multirow}
\usepackage{supertabular}

\newtheorem{theorem}{Theorem}[section]

\newtheorem{LMA}[theorem]{Lemma}

\global \count10 = 0
\global \count11 = 1
\global \count12 = 1
\global \count13 = 1
\newcommand{\com}[1]{\ifnum\count13<1 #1 \fi}
\global \count14 = 0
\def\squarebox#1{\hbox to #1{\hfill\vbox to #1{\vfill}}}
\def\qed{\hspace*{\fill}%
        \vbox{\hrule\hbox{\vrule\squarebox{.667em}\vrule}\hrule}\smallskip}
\newenvironment{proof}{\begin{trivlist}%1g263
\item[\hspace{\labelsep}{\em\noindent Proof.~}]}{\qed\end{trivlist}}
\def\squarebox#1{\hbox to #1{\hfill\vbox to #1{\vfill}}}
\def\qed{\hspace*{\fill}%
        \vbox{\hrule\hbox{\vrule\squarebox{.667em}\vrule}\hrule}\smallskip}
\newenvironment{proofsketch}{\begin{trivlist}%1g263
\item[\hspace{\labelsep}{\em\noindent Proof(sketch).~}]}{\qed\end{trivlist}}

\begin{document}

\title{
	Online interval scheduling to maximize total satisfaction
}

\author{
%	Tetsuya Araki and 
	Koji M. Kobayashi%$^{*}$ 
%	\\   
%	{\footnotesize 
%		National Institute of Informatics,
%	}
%	\\
%	{\footnotesize 
%		2-1-2, Hitotsubashi, Chiyoda-ku, 
%		Tokyo, 1018430, Japan
%	}
%	\\
%	{\footnotesize 
%		$^{*}$
%		Corresponding author:
%		kobaya@nii.ac.jp
%	}
}

\date{}

\setlength{\baselineskip}{4.85mm}
\maketitle

\begin{abstract}
	\ifnum \count10 > 0
	%
	%%%\com{%%%%%%}
	%
	%%%%%%%%%%%%%%%%%%
	%%%%%%%%%%%1%%%%%%%
	%
	%
	%%%%%%%%%%%%%%%%%%%%%%%%%%%%%%%%%%
	%
	%%%%%%%%
	{\em %%%%}%{\em %%%%}%%%%%%%%{\em %%%%}%%%%%%%%%%%%
	$m$%%%%%%%%%%%%%%
	%
	%%%%%%%%%%%%%%%%%%%%%
	%%%%%%%%%%%%%%%%%%%%%%%%%%%%%%%%%%%%%%
	%%%%%%%%%%%{\em %%}%%%%%%%%%%
	%%%%%%%%%%%%%%{\em %%}%%%%%%%%%%%%%%%%%%%%%
	%
	%
	1%%%%%%%%%%%%%%%%%%%%%%%%%%%%%%%%%%
	%
	%%%%%%%%%%%%%%%%%%%%%%%%%%%%%%%%%%%%%%%%%%%
	%%%%
	%%%%%%%%%%%%%%%%%%%%%%%%%%%%%%%%%%%%
	%%%%%%%%%%%%%%%%%%%%
	%
	%
	%%%%%%%%
	%%%%%%%%%%%%%%
	%
	%
	%%%%%%
	%%%%%%%%%%%%%%%%%%%%%%%%%%%%%%%%%%%%%%%
	%
	%%%%%%%%%%%%%%%%%%%
	%%%%%%%%%%%%%%%%
	%%%%%%%%%%%%%%%%%%%%%%%%
	%
	%
	%%%%%%%%%%%%%%%%%%%%%%%%%%
	%
	%%%%%%%%%%%%%%%
	%%%%%%%%%%%%%%%%%%%%%%%%%%%%%%%
	%%%%
	%%%%%%%%%%%
	%%%%%%%%%%%%%%%%%
	%%%%%%%%%%%%%%%%%%%%%%%%%
	%%%%%%%%%%%
	%
	%%%%%%
	$m = 2$%$m \geq 3$%%%%%
	%%%%%%%%%%%%%%%$3$%%%$4/3$%%%%%%%%%%%%%%
	%
	%%%
	%$m \geq 2$%%%%%%%%%%%%%%%%%%%%%%%%%%%%
	%
	
	%
	\fi
	\ifnum \count11 > 0
	%
	%%%\com{%%%%%}
	%
	The interval scheduling problem is one variant of the scheduling problem.  
	In this paper, 
	we propose a novel variant of the interval scheduling problem, 
	whose definition is as follows: 
	given jobs are specified by their {\em release times}, {\em deadlines} and {\em profits}. 
	An algorithm must start a job at its release time on one of $m$ identical machines,  
	and continue processing until its deadline on the machine to complete the job. 
	All the jobs must be completed and the algorithm can obtain the profit of a completed job as a user's satisfaction. 
	It is possible to process more than one job at a time on one machine. 
	The profit of a job is distributed uniformly between its release time and deadline, that is its interval, 
	and the profit gained from a subinterval of a job decreases 
	in reverse proportion to the number of jobs whose intervals intersect with the subinterval on the same machine. 
	The objective of our variant is to maximize the total profit of completed jobs. 
	This formulation is naturally motivated by best-effort requests and responses to them, 
	%%%\com{%%ed:Please clarify the meaning of "setting" in this context. %}
	%%%\com{%%setting%formulation% }
	%%%\com{It may mean "objective" or it may be referring to the environment of the scheduling problem. Please ensure that "setting" is clear throughout the manuscript.%}
	which appear in many situations. 
	In best-effort requests and responses, 
	%%%\com{%%ed:the best-effort%best-effort%}
	the total amount of available resources for users is always invariant 
	and the resources are equally shared with every user. 
	%%%\com{%%with%%%%}
	%%%\com{%%ed:the available%available%}
	%
	%
	We study online algorithms for this problem. 
	Specifically, 
	we show that for the case where the profits of jobs are arbitrary, 
	there does not exist an algorithm whose competitive ratio is bounded. 
	Then, 
	we consider the case in which the profit of each job is equal to its length, 
	that is, 
	the time interval between its release time and deadline. 
	For this case, 
	we prove that for $m = 2$ and $m \geq 3$, 
	the competitive ratios of a greedy algorithm are at most $4/3$ and at most $3$, respectively. 
	Also, 
	for each $m \geq 2$, 
	we show a lower bound on the competitive ratio of any deterministic algorithm. 
	\fi
\end{abstract}
%

%\newpage

\section{Introduction} \label{Intro}
\ifnum \count10 > 0
%
%%%\com{%%%%%%}
%
{\em interval scheduling} problem%%%%%%%%%%%%%%%%scheduling problem%1%%%%%%%
%%%%%%%%%%%%%%%%%%%%%%%%%%%%%%%%
%
%%%%%%%%%%%%%%%%%%%%%%%
%
%
$m \geq 1$%%%%%%%%%%%%%
$n \geq 1$%%%%%%%%%%%%
%%%%%{\em %%%%}%{\em %%%%}%{\em %%}%%%%%%{\em %%}%%%%%%%%%%
%%%%{\em complete}%%%%%%%%%%%%%%%%%%%%
%%%%%%%%%%1%%%%%%%%%%%%%%%%%%%%%%
%%%%%%%%%{\em %%%%}%{\em %%}%%%
%%%%%%%%%%%%%%%%%%%%%
%
1%%%%%%%%%%%%%%%%%%%%%%job%%%%%1%%%%%
%
%%%%%%%%%%complete%%%job%%%%%%%%%%%%%%%%%%
%
%
%%%%%%
%%%%%%%%%%%%%%%%%%%%%%%
%%%%%%%%%%%%%%%%%%%%%%%
%\cite{KLPS2007,KNC2007}%%%
%%%%%%%%%%%%%%%%%%%%%%%%%%%%%
%
%
%%%
%%%%%%%%%%%%%%%%%%%%%%
%
%%%%%%%%%%%%%%%%
%%%%%%%%%%%%%%%%
%%%%%%%%%%%%%%%%%%%%%%%%
%%%%%%%%%%%%%%%%%%%%%%%%%%%%%%%%%%
%
%%%%%%%%%%%%%%%%%%%%
{\em %%%%%}\cite{AB98,DS85}%%%%%%%%%%%%%%
%
%
%%%%%
%%%%%%%%%%%%%%%%%%%%%%%%
%%%%%%%%%%%$A$%%%%%%$c$%%%%%%%%%
$A$%{\em $c$-competitive}%%%%%%%
%

%
%%%%%
interval scheduling problem%%%%%%%%%%%%%
%
%
%%%%%%%%%%%%%%interval scheduling problem%%%
%%%%%%%%%%%%%%%%
%%%%%%%%
%%%%%%%%%%%
%%%%%%%%%%%%%%%%%%%%%%%%%%%%%%
%
%%%%%
1%%%%%%%%%%%%%%%%%%%%job%%%%
%%%%%%%%%%%%%%%%%%%%
%%%%%%%%%%%1%%%%%%%%%%%%%%%%%job%%%%%1%%%%%%
%
%
%%%%%
%%%%%%%%%%%%%%%%%%%%%%%%%%%%%%%%%%%%
{\em %%%%%}%%%%%%%%%%%%%%%%%%%%%%
%
%
%%%%
%%%%%%%%%%%%%%%%%
%%%%%%%%%%%%%%%%%%%%%
%%%%%%%%%%%%%%%%%%%%%%%%%
%
%
%%%%%%%%%%%%%%%%%
{\em %%%%%%%%%}%%%%%%
%%%%%%%%%%%%%%%%%%%%
%
%%%%%%%
%%%%%%%%%%%%%%%%%%%%%%
%%%%%%%%%%%%%%%%%%%%%%%%%%%%%%%%%%
%
%
%%%%%%%
%%%%%%%%%%%%%%%%%%%
%%%%%%%%%%%%%%%%%%%%%%%%%%%
%
%%%%%%
%%%%%%%%%%%%%
1%%%%%%%%%%%%%%%%%%%%job%%%%%%%%%%
%%%%job%%%%%%%%%%%%%
%
%%%%%%%%%%%%%%%%%%%%%%%
%
%
job%%%%%%%%%%%%%%%%
%
%%%%%%%%%%%%%%%%%%%%%
%%%%%%%%%%%%%%%%%%%
%%%%%%%%%%%%%%%%%%%%%%%%%%%%%%
%
%%%%%
%%%%%%%%%%%%%
%%%%%%%%%%%%%%%
%%%%%%%%%%%%%%%%%%%%%
%
%
%%%%
%%%%%%%%
%%%%%%%%%%%%%%%%%%%%%%%
%
%
%%%%%%%%%%%%%%%%%%%%%%%%%%%%
%%%%%%%%%%%%%%%%%%
%
%
%%%%%%%
%%%%%%%%%%%%%%%%%%
%%%%%%%%%%%%%%%%%%%%%%%%%%%%
%
%%%%%%%%%%%%%%%%%%%
%%%%%%%%%%%%%%%%%%%%%%%%%%%
%
%%%%%
%%%%%%%%%%%%%%%%%
%%%%%%%%%%%%%%%%%%%%%
%
%%%
%%%%%%%%%%%%%%%%%%%%
%%%%%%%%%%%%%%%%%%
%%%%%%%%%%%
%%%%%%%%%%%%%%%
%%%%%%%%%%%%%%%%%
%%%%%%%%%%%%%%%%%%%%%%%
%
%
%%%%%%%%%%
%%%%%%%%%%%%%%%%%%%%
%%%%%%%%%%%%
%
%
%%%%%%%%%%%%%
%%%%%%%%%%%%%%%%%
%%%%%%%%%%%%%%
%%%%%%%%%%%%%%%%
%
%
%%%%
%%%%%%%%%%%%%%%%%%%
%

%
\fi
\ifnum \count11 > 0
%
%%%\com{%%%%%}
%
The interval scheduling problem is one of the variants of the scheduling problem, 
%%%\com{%%ed:various variants%variants%}
which has been widely studied. 
One of the most basic definitions is as follows: 
We have $m \geq 1$ identical machines and jobs are given. 
%%%\com{%%equivalent%%%%}
A job is characterized by the {\em release time}, {\em deadline} and {\em weight} (or {\em value}). 
%%%\com{%%ed:deadline%weight%the%%%%%}
%
To {\em complete} a job, 
%%%\com{%%ed:In order%%%}
we must start to process it at its release time on a machine of the $m$ machines, 
and continue processing it until its deadline on that machine. 
%%%\com{%%ed:the%that%}
%
That is, 
the {\em processing time} (or {\em length}) of the job is the time interval between its release time and deadline. 
%%%\com{%%ed:deadline%its%%%}
%
The number of jobs which can be processed on one machine at a time is at most one. 
The objective of an algorithm is to maximize the total weight of completed jobs. 
There are many applications of the interval scheduling problem, 
such as bandwidth allocation and vehicle assignment 
(see e.g., \cite{KLPS2007,KNC2007}). 
Many variants of this problem have been proposed and extensively studied. 
Furthermore, 
research on online settings has also been considered. 
In an online variant of the interval scheduling problem, 
a job arrives at its release time and 
an online algorithm must decide whether it processes the job before the next job arrives. 
%%%\com{%%ed:or not%%%%}
%%%\com{%%ed:is given%%%%%%%%%%%%%%%%%%%%%%}
%
The performance of online algorithms is evaluated using {\em competitive analysis} \cite{AB98,DS85}. 
For any input, 
if the total weight gained by an optimal offline algorithm is at most $c$ times that gained by an online algorithm, 
the online algorithm is {\em $c$-competitive}. 
%%%\com{%%competitive%%%%%%%%%%%%%%%}
%

%
In this paper, 
we introduce a novel variant of the interval scheduling problem. 
In many existing variants of the interval scheduling problem, 
jobs (or users) require resources for an algorithm, 
%%%\com{%%ed:to%for%}
and the algorithm assigns the required resources of a machine to the job. 
Thus, 
the number of jobs assigned to one machine at a time is subject to the maximum amount of resources of the machine. 
The amount is generally one; that is, 
at most one job can be processed at a time on one machine in most variants. 
%%%\com{%%ed:Generally, Latin abbreviations are used only within parentheses. Otherwise, they are written in full in English.%}
%
%
Therefore, 
we can regard such existing variants as formulating {\em resource reservation} requests by users, 
who designate the amount of resources they want to use in advance and the responses to them. 
However, 
it is not always possible for users to designate the amount of resources they want 
when they issue requests.
Additionally, 
there are not necessarily sufficient resources of a machine to meet users' requests. 
%%%\com{%%ed:enough%sufficient%}
%
%
Thus, 
we focus on a {\em best-effort} method to manage situations, 
which is often considered paired with resource reservation methods. 
%%%\com{%%ed:deal with the%manage%}
%%%\com{%%ed:said to be%considered%}
%
%
In this method, 
the amount of resources of a machine is always invariant and 
the resources are equally shared by users 
who want to use the resources at the same time. 
%%%\com{%%ed:at a time%unclear%at the same time%}
%
%
Then, 
we formulate best-effort requests and responses to them as a variant of the interval scheduling problem. 
Specifically, 
we remove the capacity constraints from machines in our variant, 
which makes it possible to assign jobs unlimitedly on one machine at a time. 
To the best of our knowledge, 
this is the first such formulation of the interval scheduling problem. 
%%%\com{%%ed:
this is the first formulation as the interval scheduling problem. 
%It is not clear what this means in the context of the field. 
A suggestion "the first such formulation of the interval scheduling problem."
%}
%
%
Consider a given job as a user's request. 
If a machine processes the request using sufficient resources, 
the user is sufficiently satisfied with the result obtained from the process. 
%%%\com{%%ed:enough%sufficient(ly)%}
%
Conversely, 
if there are not sufficient resources to process the request, 
the user is less satisfied with the result than usual. 
%%%\com{%%ed:on the other hand%Conversely%}
%%%\com{%%ed:enough%sufficient%}
%
%
Then, 
the objective of our variant is to maximize the total satisfaction gained by users.
%%%\com{%%ed:from the%by%}
%
%
Bandwidth allocation in networks is one of the most suitable examples for best-effort requests and responses. 
%%%\com{%%%}
%
%
In this example, 
the total bandwidth which may be supplied to users on the same communication link is fixed in advance, 
and all users share the bandwidth. 
%%%\com{%%ed:total size of the%total%}
%%%\com{%%ed:all the users%all users%}
%
Hence, 
the fewer users which use the communication link at a time, 
%%%\com{%%ed:%%%is(are)%%%%}
the greater the bandwidth which each one can use, 
which means that the effective speed of the communication link is higher for the users. 
Conversely, 
the more people there are using link simultaneously, 
the lower the effective speed for each user. 
As a result, 
if the bandwidth for a user is high, 
then the user's satisfaction is high. 
Otherwise, it is low. 
Best-effort requests and responses such as bandwidth allocation could happen in many cases, 
%%%\com{%%ed:like the%such as%}
%
for example, 
the use of facilities, such as swimming pools and gyms, 
%%%\com{%%ed:uses%use%}
%%%\com{%%ed:facilities, such ... ,%}
passenger trains without reservations, and 
buffet style meals. 
%%%\com{%%%%%%%}
%%%\com{%%%%%train%%%%%%%}
%
%
Therefore, 
we have sufficient incentives to study our variant. 
\fi
\ifnum \count10 > 0
%
%%%\com{%%%%%%}
%
{\bf %%%%%}~
%
%%%%%
%%%%%%%%%interval scheduling problem%%%%%
%%%%%%%%
%
%%%%%%%%
%%%%%%%%%%%%%%%
%%%%%%
%%%%%%%%%%%
%
%
%%%%
%%%%%%%%%%%%%%%
%
%%%%%%%%%%
%job%%%%%%%%%%%%%%job%%%%%%%%%%%%%%%%%%%%%%%%%%%%%%%%%%
%%%%%%%%%%%%%%%%%%%%%%%%%%%%%
%
%%%%%%%%%%%%%%%%%
%job%%%%%%%%%%%%%%job%%%%%%%%%%
{\em uniform profit case}%%%%
%%%%%
%
%
%%%%%%%
1%%%%%%%%%%%%%%%%%%%%%%%%%%%%%
%
%%%%
%%%%%%%%%%%%%%%%%%%%%%%%%%%%
%
%
%%%
%%%%%%%%
%%%%%%%%$GR$%%%%%%%%%
$GR$%%%%%%%%%%%%%%%%%%%%%%%%%%%
%%%%%%%%%%%%%%%%%%
%
$m \geq 3$%$m = 2$%%%%
$GR$%%%%%%%%%%%$2$%%%$4/3$%%%%%%%%%
%
%%%
$m=2$%%%%
$GR$%%%%%%%%%%$4/3$%%%%%%%%%
%
%%%%$m$%%%%%
%%%%%%%%%%%%%%%%%%%%%
%~\ref{tab:OurRes}%\ref{sec:lbuc}%%%~\ref{tab:thm:VL.LBm}%%%%%%%%%%%
%

%
\fi
\ifnum \count11 > 0
%
%%%\com{%%%%%}
%
{\bf Our Results.}~
In this paper, 
we propose and analyze a novel variant of the interval scheduling problem. 
%
%%%%%%%%
%%%%%%%%%%%%%%%
%%%%%%
%%%%%%%%%%%
%
%
We study online algorithms for this problem. 
Specifically, 
in the case where the profits of jobs are arbitrary; 
that is, 
the profits are not relevant to the lengths of jobs, 
we show that the competitive ratio of any deterministic algorithm is unbounded. 
Then, 
we introduce the profits of jobs are equal to their lengths, 
which is a more natural case, 
called the {\em uniform profit case}. 
In this case, 
the total amount of time during which at least one job is scheduled on a machine is equal to the total amount of the satisfaction gained on the machine. 
That is, 
the objective of this case can be regarded as maximizing the working hours of all the machines.  
We analyze the performance of a greedy algorithm $GR$ in this case. 
Since $GR$ is a significant algorithm from a practical point of view, 
it is worthwhile to evaluate its performance. 
When $m = 2$ and $m \geq 3$, 
we show that the competitive ratios of $GR$ are at most $4/3$ and at most $3$, respectively. 
When $m=2$, 
we prove that a lower bound on the competitive ratio of $GR$ is $4/3$. 
That is, 
for $m = 2$, our analysis of $GR$ is tight. 
Also, 
we show lower bounds of any deterministic online algorithms for each $m \geq 2$,  
which are summarized in Table~\ref{tab:OurRes} and Table~\ref{tab:thm:VL.LBm} in Sec.~\ref{sec:lbuc}. 
\fi

\begin{table*}
\begin{center}
	\renewcommand{\arraystretch}{1.5}
	%\caption{%%%%%}
	\caption{Our Results}
	\begin{tabular}{l|l|l}
		\hline
        %        $m$ & %% & %% \\
			$m$ & Upper bound & Lower bound \\
		\hline
			2 & $4/3 \leq 1.334$ & $(10 - \sqrt{2}) / 7 \geq 1.226$ \\
		\hline
        	3 & \multirow{5}{50pt}{ $3$ } & $7/6 \geq 1.166$ \\ %1.1666666666
		\cline{1-1}
		\cline{3-3}
    		4 && $(22-2\sqrt{2})/ 17 \geq 1.127$ \\ %1.1277395809
		\cline{1-1}
		\cline{3-3}
			5 && $(420-15\sqrt{7})/ 333 \geq 1.142$ \\ %1.14208327428
		\cline{1-1}
		\cline{3-3}
			6 && $(51-6\sqrt{2})/ 41 \geq 1.140$ \\
		\cline{1-1}
		\cline{3-3}
			$\infty$ && $(48 - 2 \sqrt{2})/41 \geq 1.101$ \\
		\hline
	\end{tabular}
	\label{tab:OurRes}
\end{center}
\end{table*}

\ifnum \count10 > 0
%
%%%\com{%%%%%%}
%
{\bf %%%%}~
interval scheduling problem%%%%%%%%%%%%%%%%%
%
%%%%%%%%%
Arkin%Silverberg~\cite{AS1987}%Bouzina%Emmons~\cite{BE1996}%%
%%%%%%%%%%%%%%%%%%%%%%%%%%%%%%%%%%%%%%%%%%%%%%%%
%

%
%%%%%%%%%%%%%%%%%%%%%%%%%%%%
%
%
%%%%%%%%%%%%%%%%%%%%%%%{\em preempt}%%%%
%
preempt%%%%%%%%%%%
complete%%%%%%%%%%%%%%%%%%%%%%%
Faigle and Nawijn~\cite{FN1995}
1-competitive%%%%%%%%%%%%%%%%%%%%%%
Carlisle%Lloyd~\cite{CL1995}%%%%%%%%%%%%%%%%%%%%%%%%%%%%%
preempt%%%%%%%%%%%
complete%%%%%%%%%%%%%%%%%%%%%%%%
%%%%%%1%%%%%%%%%%
Woeginger~\cite{GJW1994}%%%%%%%%%%%%%%%%%%%%%%%%%%%%%%%
Canetti%Irani~\cite{CI1998}%%
%%%%$O(\log \Delta)$%%%%%%%%%%%%%%%%%%%%%
%%%%%$\Omega( \sqrt{ \log \Delta / \log \log \Delta } )$%%%%%%%%%%%%
%%%%
$\Delta$%%%%%%%%%%%%%%%%%%%%%%%%%%%%
%
%%%%%%
%%%%%%%%%%%%
%%%%%%%%%%%%%%%%%%%%%%%%%%%%%%
%
%
%%%%
%%%%%%%%%%%%%%%%%%%%%
%%%%%%%%%%%%%%%%%%%%%%%%%%%%%%%%%%%%%
%%%%%%%%%%%%%
%
%%%%%%%%%%%
Woeginger~\cite{GJW1994}%%%%%%%4%%%%%%%%%%%%%%%%
%%%%%%%%%%%%%%%
%
%
%%%%%%%%%%%%%%%%%%%%%%%%%~\cite{SS1998,ME2004,FPZ2008,EL2010,FPY2012,FPZ2014}%
%
$m = 1$%%%
%%%%%%%%%%
Fung%~\cite{FPZ2014}%$2$%%%%%%%%%
Epstein%Levin~\cite{EL2010}%$1 + \ln 2 \geq 1.693$ %1.69314718056
%%%%%
%
%
$m \geq 2$%%%%
%%%%%%%%%%%%%%%%%%%%%%%%
Fung%~\cite{FPY2012}%
$m$%%%%%%%$2$%%%%
$m$%%%%%%%%$2 + 2/(2m-1)$%%%%%%%%%%
%
%%%%%%%%%%%%%
$m = 2$%%Fung%~\cite{FPZ2008}%$2$%%%%%%%
%%%\com{%%\cite{FPZ2008}%%%%non-pre%%%%\cite{FPY2012}%%%%pre%%%%%%%%%%%%%}
$m \geq 3$%%%
Epstein%Levin~\cite{EL2010}%%%%$1 + \ln 2 \geq 1.693$
%%%%%%%%%%%%%
Fung%~\cite{FPY2012}%%%%%%%%
preemption%%%%%%%%%%%%
Lipton%Tomkins~\cite{LT1994}%%
%%%%$O((\log \Delta)^{1+\epsilon})$%%%%%%%%%%%%%%%%
%%%%%%%%%%%%%%%%$\Omega(\log \Delta)$%%%%%%%%%%
%

%
%%%%%%%%%%%%%%%%%%%
%%%%%%%%%%%%%%%%%%%%%%%%%%%%%%%%
%
%%%
%%%%%%%%%%%%%%%%%%%
%%%%%%%
{\em slack} $\varepsilon > 0$%%%%%%%%%%%%%%%
%%%%%%%%%
%%%%%%%%%%%%%%%%%%%$x$%%%%%%%
%%%%
$x = 1/(1 + \varepsilon)$
%%%%%%%%
slack%%%%%%%%%%%%%%%%%%%%%
%
%
%%%%%%%
preemption%%%%%%%%
%%%%%%%%%%%%%%%%%%%%%%%%%%%%%%%%%%%%
%
%
%%%%%%%%%%%%%%%%%%%%%%%%%%%%%%%%
%%%%%%$1 + 1/\varepsilon$%%%
%%%%%%%\cite{BH1997,DP2001,SZ2016}%
%

%

%
\fi
\ifnum \count11 > 0
%
%%%\com{%%%%%}
%
{\bf Related Results.}~
Much research on the interval scheduling problem has been conducted. 
%%%\com{%%a lot of%much%}
%%%\com{%%done%conducted%}
%
Arkin and Silverberg~\cite{AS1987} and Bouzina and Emmons~\cite{BE1996} provided polynomial time algorithms to solve the interval scheduling problem. 
There is also much research on online interval scheduling problems. 
If an online algorithm aborts a job $J$ which was placed on a machine, 
%%%\com{%%into%on%}
then we say that the algorithm {\em preempts} $J$. 
In the case in which preemption is allowed, 
Faigle and Nawijn~\cite{FN1995} designed a 1-competitive algorithm to maximize the number of completed jobs. 
This algorithm was independently discovered by Carlisle and Lloyd~\cite{CL1995} but used only for the offline setting. 
Moreover, 
for the variant in which the objective is to maximize the total weight of completed jobs, 
Woeginger~\cite{GJW1994} showed that no any competitive deterministic algorithm exists (even) for $m=1$. 
%%%\com{%%there does not exist%%%%%%%%exists%}
%
%
Canetti and Irani~\cite{CI1998} provided a randomized online algorithm whose competitive ratio is $O(\log \Delta)$ 
%%%\com{%%ed:showed%%not clear%}
and proved that a lower bound on the competitive ratio of any randomized algorithm is $\Omega( \sqrt{ \log \Delta / \log \log \Delta } )$, 
where $\Delta$ is the ratio of the longest length to the shortest length. 
This result indicates that the competitive ratio of an online algorithm may become worse depending on a given input 
even if it is supported by randomization. 
Additionally, 
the setting in which the jobs are unit length has been extensively studied. 
%%%\com{%%ed:lengths of jobs are unit has%lengths of jobs are unit has%}
%
For the one machine setting, 
Woeginger~\cite{GJW1994} designed a deterministic algorithm 
whose competitive ratio is at most $4$ and showed that this is the best possible ratio. 
%%%\com{%%ed:it%this%}
%
%
There has also been much work regarding randomized algorithms (e.g. \cite{SS1998,ME2004,FPZ2008,EL2010,FPY2012,FPZ2014}).
When $m=1$, 
the current best upper and lower bounds on the competitive ratios of randomized algorithms are $2$ by Fung et~al.~\cite{FPZ2014} and $1 + \ln 2 \geq 1.693$ by Epstein and Levin~\cite{EL2010}, respectively. %1.69314718056
For $m \geq 2$, 
Fung et~al.~\cite{FPY2012} proved that, 
if $m$ is even, an upper bound is $2$, 
and otherwise $2 + 2/(2m-1)$. 
However, 
for $m = 2$, 
the current best lower bound is $2$ by Fung et~al.~\cite{FPZ2008}. 
When each $m \geq 3$, 
Fung et~al.~\cite{FPY2012} indicated that 
we can obtain a lower bound of $1 + \ln 2 \geq 1.693$ in a similar manner to the lower bound of Epstein and Levin~\cite{EL2010}. 
%%%\com{%%ed:pointed out%indicated%}
%%%\com{%%ed:by%of%}
%
%%
%
If preemption is not allowed, 
Lipton and Tomkins~\cite{LT1994} proposed a randomized algorithm whose competitive ratio is $O((\log \Delta)^{1+\epsilon})$ 
and proved that a lower bound of any randomized algorithms is $\Omega(\log \Delta)$. 
For a job given in the interval scheduling problem, 
its length is equal to the length of the time between its release time and deadline. 
On the other hand, 
a variant in which the job length is generalized has also been studied. 
Specifically, 
a parameter {\em slack} $\varepsilon > 0$ is introduced, 
whose value is known to an algorithm in advance, 
and the length of a job is at most $x$ times as long as the length of the time between its release time and deadline, 
in which $x = 1/(1 + \varepsilon)$. 
In this variant, 
preemption is allowed and 
to complete a job, 
an algorithm must process it during its length by its deadline after its release time. 
For several $m$, 
optimal online algorithms were designed \cite{BH1997,DP2001,SZ2016}, 
whose competitive ratios are $1 + 1/\varepsilon$. 
\fi
%

%%%%%%%%%%%%%%%%%%%%%%%%%%%%%%%%%%%%%%%%%%%%%%%%%
\section{Model Description} \label{sec:model}
\ifnum \count10 > 0
%
%%%\com{%%%%%%}
%
$m (\geq 2)$%%%%%%%%%%
{\em %%}%%%$n (\geq 1)$%%jobs%%%%%%%%%%%%%%%
%
%
%job $J$%%3%%$(r, d, v)$%%%%%%%%%%
%
$r(J)$%$J$%{\em %%%%}%
$d(J)$%$J$%{\em %%%%}%
$v(J)$%$J$%{\em %%}%%%%
%
%
%%%%%%$ALG$%%job%$m$%%%%%%%%1%%%%%%%%%%%%%%%
%
1%%%%%%%%%%%%%%%%%%%%%%%%%%%%%%%%%%
%
%%%%%%%%%%%%%%%%%%%%%%%%%%%%%%%%%%%%%%%%%%%
%%%%
%%%%%%%%%%%%%%%%%%%%%%%%%%%%%%%%%%%%
%%%%%%%%%%%%%%%%%%%%
%
%
%%%%%%%%%%%%%%%
%
%%%%%%%
$ALG$%$a (\in [1, m])$%%%%%%%%%%%
%%$(x, y) \hspace{1mm} (x < y)$%%
%%%2%%%%%%%%%%%%%%%%%%%
%%%2%%%%%%%%%%%%%%%%%%%%%%%%
%%%%%$ALG$%$a$%%%%%%%{\em $P$-interval}%%%%
%
%
%%%%%%%%%%%%%$k_{ALG}(a, x, y)$%%%%
%
%
$ALG$%%%%%$J$%$i$%%%%%%%%%%%%%%%
$m_{ALG}(J) = i$%%%%%%
%
%
%%%%%%$ALG$%%%%$J$%%%%%
$ALG$%$m_{ALG}(J)$%%%%%%%%%%%
%%$(r(J), d(J))$%$b (\geq 1)$%%$P$-interval $(x_{i}, x_{i+1}) \hspace{1mm} (i = 1, \ldots, b-1)$%%%%%%%%%%%%%
%%%%
$r(J) = x_1 < x_2 < \cdots < x_b = d(J)$
%%%%%%
%
%
%%%%%
$ALG$%%%$J$%%%$[x_{i}, x_{i+1}]$%%%%%%%%%%%%%
\[
	V_{ALG}(J, i) = \frac{ x_{i+1} - x_{i} }{ d(J) - r(J) } \frac{ v(J) }{ k_{ALG}(m_{ALG}(J), x_{i}, x_{i+1}) } 
\]
%%%%%%
%%~\ref{fig:modelex}%%%%
%
%
%%%
$ALG$%%%job $J$%%%%%%%%%%
\[
	V_{ALG}(J) = \sum_{ i = 1 }^{ b-1 } V_{ALG}(J, i)
\]
%%%%%%
%
%
%%%
%%$\sigma$%%%%$ALG$%{\em %%}%%
%%%%%$n$%%jobs%%%%%%%%%%%
%%%%%
\[
	V_{ALG}(\sigma) = \sum_{ J \in {\cal L} } V_{ALG}(J)
\]
%%%%%%
%%%%${\cal L}$%$n$%%%%%%%%%%%%%%%%
%
%
%%%%%%%%$n$%%jobs%%%%%%%%%%%%%
%

%
%%%%%
%%%%%%%%%%%%%%%%
%
%%%%%%
$n$%%jobs%1%%%%%%%%%
%
%
%%%%%%%%%%%%%%%%%%%%%%%%%%%%%%
%
%
%%%%%%%%%%%%%%%%%%job%%
%%job%%%%%%%%%%%%%%%%%%%%%%%%%%%%%%%
%
%
%%%%%%%%%%%%%%%%%%%%%%%%%
%%%%%
%%%%%%%%%%%%%%%
%
%
%%%%%job%%%$n$%%%%%%%%%%%%%%%%%%%%%%
%%%job%%%%%%%%%%%%%%
%
%
%%%%%$\sigma$%%%%%
%%%%%%%%%%%%%%$OPT$%%%%%%
%%%%%%%%%%%$A$%%%%%%$c$%%%%%%%%%
$A$%%%%%%%$c$%%%%
%%%%$A$%{\em $c$-competitive}%%%%%%%%
%%%\com{%%%%%$OPT$%%%%}
%

%
\fi
\ifnum \count11 > 0
%
%%%\com{%%%%%}
%
We have $m (\geq 2)$ identical machines. 
A list consisting of $n (\geq 1)$ jobs is provided as an {\em input}. 
%%%\com{%%ed:given%provided%}
%
%
A job $J$ is specified by a triplet $(r, d, v)$, where
%%%\com{%%three-tuple $<$ triplet%}
%%%\com{%%ed: :%,where%}
%
$r(J)$ is the {\em release time} of $J$, 
$d(J)$ is the {\em deadline} of $J$, 
and 
$v(J)$ is the {\em profit} of $J$. 
An algorithm $ALG$ must place each job onto one of the $m$ machines. 
It is possible to place more than one job at a time on one machine. 
The profit of a job is distributed uniformly between its release time and deadline, that is its interval, 
and the profit gained from a subinterval of a job decreases 
in reverse proportion to the number of jobs whose intervals intersect with the subinterval on the same machine. 
Specifically, 
the profit from the subinterval is defined as follows: 
For an algorithm $ALG$, 
if the numbers of jobs placed at any two points in an interval $(x, y) \hspace{1mm} (x < y)$ are equal on $ALG$'s $a (\in [1, m])$th machine 
and $(x, y)$ does not contain any endpoint of the interval of a job placed on the machine 
after processing of the input, 
then we call the interval a {\em $P$-interval} on $ALG$'s $a$th machine. 
Also, 
let $k_{ALG}(a, x, y)$ denote the number of the jobs. 
If an algorithm $ALG$ places a job $J$ onto the $a$th machine, 
then we define $m_{ALG}(J) = a$. 
%%%\com{%%%%%%%}
%
%
For an algorithm $ALG$ and a job $J$, 
suppose that the interval $(r(J), d(J))$ consists of $b (\geq 1)$ $P$-intervals $(x_{i}, x_{i+1}) \hspace{1mm} (i = 1, \ldots, b-1)$ on $ALG$'s $m_{ALG}(J)$th machine 
such that 
$r(J) = x_1 < x_2 < \cdots < x_b = d(J)$. 
Then, 
we define the satisfaction ({\em profit}) which is yielded from $[x_{i}, x_{i+1}]$ of $J$ and $ALG$ gains as 
%%%\com{%%%%%%%%%%%%%%%%}
\[
	V_{ALG}(J, i) = \frac{ x_{i+1} - x_{i} }{ d(J) - r(J) } \frac{ v(J) }{ k_{ALG}(m_{ALG}(J), x_{i}, x_{i+1}) } 
\]
(see an example in Fig.~\ref{fig:modelex}). 
We define the satisfaction ({\em profit}) of $J$ gained by $ALG$ as  
\[
	V_{ALG}(J) = \sum_{ i = 1 }^{ b-1 } V_{ALG}(J, i). 
\]
The {\em profit} of $ALG$ for an input $\sigma$ is defined as 
\[
	V_{ALG}(\sigma) = \sum_{ J \in {\cal L} } V_{ALG}(J), 
\]
where ${\cal L}$ is a list consisting of the $n$ given jobs.
The objective is to maximize the total satisfaction of the $n$ jobs. 
%

%
%%%\com{%%NPhard%%%%%%%%%%%%ESA%%%%%%}
In this paper, 
we consider an online variant of this problem. 
Specifically, 
$n$ jobs are given one by one. 
The jobs are not necessarily given in order of release time. 
%%%\com{%%%%%%%%%}
%
%
An online algorithm must place a given job to a machine before the next job is given. 
Once a job is placed on a machine, 
it cannot be removed later. 
That is, 
preemption is not allowed. 
The total number $n$ of given jobs is not known to the online algorithm, and 
it does not require this information until after all the jobs arrive. 
We say that 
the competitive ratio of an online algorithm $A$ is at most $c$ or 
$A$ is $c$-competitive 
if, for any input, 
the profit gained by an offline optimal algorithm $OPT$ is at most 
$c$ times the profit gained by $A$. 
\fi
\ifnum \count12 > 0
\begin{figure*}[ht]
	 \begin{center}
	  \includegraphics[width=100mm]{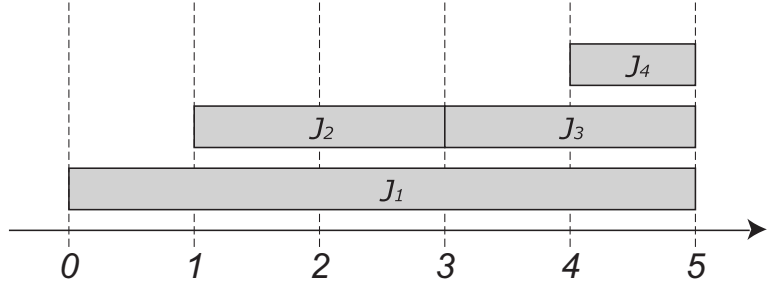}
	 \end{center}
	 \caption{
\ifnum \count10 > 0
%
%%%\com{%%%%%%}
%
4%%%%%$J_{i} \hspace{1mm} (i \in [1, 4])$%%%%%%%%%%%%%%%%%%%%%%%
%%$[0, 5]$%%
4%%$P$-interval $[0, 1], [1, 3], [3, 4], [4, 5]$%%%%%%%%%
%
%%%%
$J_{1}$%%%$[0, 5]$%%%%%%%%
$J_{1}$%%%$[0, 1]$%%%%%%%$v(J_{1})/5$%%%%
%
%%%
%%$[4, 5]$%%%
3%%%%%$J_{1}, J_{3}, J_{4}$%%%%%%%%
%%%%%%%%%%%%%%%%%%%%
$v(J_{1})/15, v(J_{3})/6, v(J_{4})/3$
%%%%
%
\fi
\ifnum \count11 > 0
%
%%%\com{%%%%%}
%
An example in which four jobs $J_{i} \hspace{1mm} (i \in [1, 4])$ are placed on the same machine. 
The interval $[0, 5]$ consists of four $P$-intervals $[0, 1], [1, 3], [3, 4]$ and $[4, 5]$. 
For example, 
$J_{1}$ exists during the interval $[0, 5]$ and thus 
the profit from $[0, 1]$ of $J_{1}$ is $v(J_{1})/5$. 
Also, 
since three jobs $J_{1}, J_{3}$ and $J_{4}$ in $[4, 5]$, 
the profits from this interval of $J_{1}, J_{3}$ and $J_{4}$ are $v(J_{1})/15, v(J_{3})/6$ and $v(J_{4})/3$, respectively. 
\fi
			}
	\label{fig:modelex}
\end{figure*}
\fi
%
%%%%%%%%%%%%%%%%%%%%%%%%%%%%%%%%%%%%%%%%%%%%%%%%%
\section{General Profit Case} \label{sec:general}

%Due to page limitations, 
%we omit some of the proofs of the following lemma and theorems. 
%They are included in Appendix~\ref{ap.sec:proofs}.
%The full version of this paper is available at\\
%{\tt \footnotesize https://drive.google.com/drive/folders/0B67SJoEoHQkNVGxFYkloV1J4M1k?usp=sharing}. 

%
\ifnum \count10 > 0
%
%%%\com{%%%%%%}
%
%%%%%
job%%%%%%%%%%%%%%%%
$m \geq 3$%%%%%%%%%%%%%%%$m = 2$%%%%%%%%
\fi
\ifnum \count11 > 0
%
%%%\com{%%%%%}
%
In this section, 
we consider the case in which the profits of jobs are arbitrary.
First, we consider the case $m=2$ for better understanding of any $m \geq 3$. 
\fi
\begin{theorem}\label{thm:GLB2}
	\ifnum \count10 > 0
	%
	%%%\com{%%%%%%}
	%
	$m = 2$%%%%
	%%%%%%%%%%%%%%%%%%%%%%%
	%%%%%%%
	%
	
	%
	\fi
	\ifnum \count11 > 0
	%
	%%%\com{%%%%%}
	%
	When $m = 2$, 
	there does not exist any deterministic online algorithm whose competitive ratio is bounded. 
	\fi
\end{theorem}
\begin{proof}
	\ifnum \count10 > 0
	%
	%%%\com{%%%%%%}
	%
	%%%%%%%%%%%$ON$%%%%%%%%
	%
	%%%%%%%%%%%%%%job%%%%1%%%%
	%
	%%%$ON$%%%%%%%%%%%%%%%%%%%%%%%%
	%
	%%
	%%%%%%%%%%%%%%%%%%%%%%%%$z(>0)$%%%%%%%%
	%%%%%%%%%%%%%%%%$J$%%%%%
	$J$%%%%%%%%%%%%%%%$v(J)/(z+1)$%%%%
	%%%%%$ON$%2%%%%%%%%%%$z$%%%%%%%%%%%%%%%%%
	%
	%%%$OPT$%$J$%%%1%%%%%%%%%%%%
	%%$v(J)$%%%%%%%%%%
	%
	
	%
	\noindent\vspace{-1mm}\rule{\textwidth}{0.5mm} 
%	\vspace{-3mm}
%	{\bf Routine}\\
%	\rule{\textwidth}{0.1mm}
	%
	%%%
	%
	%\noindent
	{\bf\boldmath Step~1:} 
	${a}_1 > 0$%${\tt MaxC} > 0$%%%%%%%%%%%%
	$\ell := 1$, 
	$s_1 := 0$.\\
	%
	%%%%
	%
	%\noindent
	{\bf\boldmath Step~2:} 
	$2 \cdot {\tt MaxC}$%%job $J$%%%%%
	%%%%$r(J) = 0, d(J) = 1$%%%%
	$ON$%%%$2 \cdot {\tt MaxC}$%%1st(2nd)%%%%%%%%%job%%%$x_1$($x_2$)%%%%
	%%%%%%%%%%%
	$x_1 \geq x_2$%%%%
	\\
	%
	%%%%
	%
	%\noindent
	{\bf\boldmath Step~3:} 
	$\ell = {\tt MaxC} + 1$%%%%%%%%%%%
	\\
	%
	%%%%
	%
	%\noindent
	{\bf\boldmath Step~4:} 
	${a}_{\ell}$%%job $J_{\ell,i} \hspace{1mm} (i = 0, 1, \ldots, {a}_{\ell}-1)$%%%%%
	%%%%
	$r(J_{\ell, i}) = s_{\ell} + p_{\ell} i/{a}_{\ell}$, 
	$d(J_{\ell, i}) = s_{\ell} + p_{\ell} (i+1)/{a}_{\ell}$
	%%%%
	%%%%
	$p_{1} = 1$
	$\ell \geq 2$%%%%%
	$p_{\ell} = d(J_{\ell-1, 1}) - r(J_{\ell-1, 1})$
	%%%%\\
	%
	%%%%
	%
	%\noindent
	{\bf\boldmath Step~5:} 
	%
	%%%2%%Case%%%%%%%%%
	\\
	%
	%%%%
	%
	%\noindent
	\hspace*{1mm}
	{\bf\boldmath Case~5.1 ($ON$%$J_{\ell, i}$%1%%2nd%%%%%%%%%%%):} 
	%
	%%%%%%%
	\\
	%
	%%%%
	%
	%\noindent
	\hspace*{1mm}
	{\bf\boldmath Case~5.2 (%%%%%%%):} 
	%
	%%%%%$s_{\ell}$%%%%%%
	$ON$%2nd%%%%%%%%%job%$J_{\ell, i'}$%%%%
	$s_{\ell+1} := s_{\ell} + p_{\ell} i'/{a}_{\ell}$%
	$(2 + \sum_{j = 1}^{\ell} \max \{ \frac{1}{x_1 + 1}, \frac{1}{x_2 + j} \} {a}_{j}) / {a}_{\ell+1}$
	%%%%%%%%%%%%%%${a}_{\ell+1} > 0$%%%%
	$\ell := \ell + 1$%
	Step~3%%
	\\
	\noindent\vspace{-1mm}\rule{\textwidth}{0.5mm} 
	%
	
	%
	%%%%%%%%%%%%$\ell$%%%%
	{\tt FinC}%%%%
	$ON$%%%%%%%%%$\sigma_{{\tt FinC}}$%%%%
	%
	%
	%%%%%%%$\sigma_{{\tt FinC}}$%%%%%
	%%%%%%%%%%%%%$OFF$%%
	%%%%%%%%%${\tt FinC}$%%%Step~4%%%%%%%
	%%%%%${a}_{{\tt FinC}}$%%job $J_{{\tt FinC}, i}$%%%2nd%%%%%%%%%
	%%%%%%%%job%1st%%%%%%%%%%
	%
	%
	%%%%${\tt FinC} \leq {\tt MaxC}$%%%%%
	\[
		V_{OPT}(\sigma_{{\tt FinC}}) \geq V_{OFF}(\sigma_{{\tt FinC}}) \geq {a}_{{\tt FinC}}
	\]
	%%%%%
	\[
		V_{OPT}(\sigma_{{\tt MaxC}+1}) \geq V_{OFF}(\sigma_{{\tt MaxC}+1}) \geq {a}_{{\tt MaxC}}
	\]
	%%%%%%
	%
	%
	%%%$ON$%%%%%%%%%%%%%
	%
	$ON$%Step~2%%%%%
	%$j \in \{ 1, 2 \}$%%%%%
	$j$th %%%%%%%%%%%job%%%$x_j$%%%%%
	1%%job%%%%%%$1/x_j$%%%%
	%
	%%%%%%%%%%%%%%%
	%%%%%%Step~2%%%%%%%%%job%%%%%job%%%%%%%%%$2$%%%%
	%
	%
	
	%
	$\ell (\in [2, {\tt FinC}-1])$%%%%%%%Step~4%%%%%%%%%%%%%%
	Case~5.2%%%%%%
	$ON$%$\ell-1$%%%Step~4%%%%%%%%%%%%%%%%%%%%%%
	%%~\ref{fig:glb2}%
	%
	%
	%%%%
	$\ell$%%%Step~4%%%%%%%%
	$ON$%2nd%%%%%
	$[s_{\ell}, s_{\ell} + p_{\ell}]$
	$x_2 + \ell - 1$%%job%%%%%%%%%%%
	%
	%
	%%%%
	Case~5.2%%%%%%%%
	$\ell$%%%Step~4%%%%%%%%%1%%job%%%%%
	$\ell$%%%Step~4%%%%%%%%%%%%
	%%$\max \{ \frac{1}{x_1 + 1}, \frac{1}{x_2 + \ell} \}$
	%%%%
	%
	%%%%
	$\ell$%%%Step~4%%%%%%%%%%%%%%%%%%%%%%%%%
	%%$\max \{ \frac{1}{x_1 + 1}, \frac{1}{x_2 + \ell} \} {a}_{\ell}$
	%%%%
	%
	%
	%%%
	Case~5.1%%%%%%%%
	$\ell$%%%Step~4%%%%%%%%%1%%job%%%%%
	%%$\frac{1}{x_1 + 1}$
	%%%%
	%
	%%%%
	%%%%%
	%%$\frac{{a}_{\ell}}{x_1 + 1}$
	%%%%
	%
	%%%
	%%%%%
	$x_1 \geq x_2$%$x_1 + x_2 = 2{\tt MaxC}$%%%%%%%%
	$x_1 \geq {\tt MaxC}$%%%%%%
	%
	%
	%%%%%
	${\tt FinC} \leq {\tt MaxC}$%Case~5.1%%%%%%%%%%%%%%%
	\begin{eqnarray*}
		V_{ON}(\sigma_{\tt FinC}) 
			&\leq& 2 + \sum_{j = 1}^{{\tt FinC}-1} \max \left \{ \frac{1}{x_1 + 1}, \frac{1}{x_2 + j} \right \} {a}_{j}
						+ \frac{{a}_{{\tt FinC}}}{x_1 + 1} \\
			&<& 2 + \sum_{j = 1}^{{\tt FinC}-1} \max \left \{ \frac{1}{x_1 + 1}, \frac{1}{x_2 + j} \right \} {a}_{j}
						+ \frac{{a}_{{\tt FinC}}}{{\tt MaxC}}
	\end{eqnarray*}
	%%%%%%
	%%%%
	2%%%%%%%$x_1 \geq {\tt MaxC}$%%%%%%%
	%
	%
	%%%
	Step~3%%%%%%%%%%%%
	%%%%%${\tt FinC} = {\tt MaxC}+1$%%%%%%%%
	\begin{eqnarray*}
		V_{ON}(\sigma_{{\tt MaxC}+1}) 
			&\leq& 2 + \sum_{j = 1}^{{\tt MaxC}} \max \left \{ \frac{1}{x_1 + 1}, \frac{1}{x_2 + j} \right \} {a}_{j}\\
			&\leq& 2 + \sum_{j = 1}^{{\tt MaxC}-1} \max \left \{ \frac{1}{x_1 + 1}, \frac{1}{x_2 + j} \right \} {a}_{j}
						+ \frac{{a}_{{\tt MaxC}}}{{\tt MaxC}}
	\end{eqnarray*}
	%%%%%%
	%
	%
	%%%%%${\tt FinC} \leq {\tt MaxC}$%%%%
	\begin{eqnarray*}
		\frac{ V_{OPT}(\sigma_{{\tt FinC}}) }{ V_{ON}(\sigma_{{\tt FinC}}) } 
			&\geq& \frac{ {a}_{{\tt FinC}} }{ 2 + \sum_{j = 1}^{{\tt FinC}-1} \max \{ \frac{1}{x_1 + 1}, \frac{1}{x_2 + j} \} {a}_{j}  +  {a}_{{\tt FinC}} / {\tt MaxC} }\\
			&\geq& \frac{ {a}_{{\tt FinC}} }{ {a}_{{\tt FinC}} \delta_{{\tt FinC}} +  {a}_{{\tt FinC}} / {\tt MaxC} }
			\geq \frac{ 1 }{ \delta_{{\tt FinC}} + 1 / {\tt MaxC} }
	\end{eqnarray*}
	%%%%%%
	%%%%
	$\delta_{{\tt FinC}} = (2 + \sum_{j = 1}^{{\tt FinC}-1} \max \{ \frac{1}{x_1 + 1}, \frac{1}{x_2 + j} \}) / {a}_{{\tt FinC}}$
	%%%%%
	%
	%
	Case~5.2%$\delta_{{\tt FinC}}$%%%%%%
	$\delta_{{\tt FinC}}$%$1/{\tt MaxC}$%%%%%%%%%%%
	$1/(\delta_{{\tt FinC}} + 1/{\tt MaxC})$%%%%%%%%%
	%
	%
	%%%%
	\[
		\frac{ V_{OPT}(\sigma_{{\tt MaxC}+1}) }{ V_{ON}(\sigma_{{\tt MaxC}+1}) } 
			\geq {\tt MaxC}
	\]
	%%%%%%
	%
	
	%
	\fi
	\ifnum \count11 > 0
	%
	%%%\com{%%%%%}
	%
	Consider a deterministic online algorithm $ON$. 
	Let the profits of all the given jobs in this proof be one. 
	First, 
	we outline the routine to provide $ON$ with an input. 
	%%%\com{%%ed:sketch%outline%}
	%
	If there exist at least $z(>0)$ jobs placed during a time interval on each machine, 
	and a new job $J$ is included in the interval, 
	then $ON$ can obtain a profit of at most $v(J)/(z+1)$ from $J$. 
	%%%\com{%%ed:the profit%a profit%}
	The following routine attempts to force each $ON$'s machine to place at least $z$ jobs during an interval. 
	%%%\com{%%ed:tries%attempts%}
	%
	$OPT$ places only $J$ onto one machine and can obtain the profit of $v(J)$. 
	\noindent\vspace{-1mm}\rule{\textwidth}{0.5mm} 
%	\vspace{-3mm}
%	{\bf Kierstead-Trotter algorithm}\\
%	\rule{\textwidth}{0.1mm}
	%
	%%%
	%
	%\noindent
	{\bf\boldmath Step~1:} 
	Let ${a}_1 > 0$ and ${\tt MaxC} > 0$ be sufficiently large integers. 
	%%%\com{%%ed:enough%sufficiently%}
	%
	$\ell := 1$ and $s_1 := 0$. \\
	%
	%%%%
	%
	%\noindent
	{\bf\boldmath Step~2:} 
	Give $2 \cdot {\tt MaxC}$ jobs $J$ such that $r(J) = 0$ and $d(J) = 1$. 
	Let $x_1$ ($x_2$) denote the number of the $2 \cdot {\tt MaxC}$ jobs which $ON$ places onto the first (second) machine. 
	Without loss of generality, 
	we assume that $x_1 \geq x_2$. 
	\\
	%
	%%%%
	%
	%\noindent
	{\bf\boldmath Step~3:} 
	If $\ell = {\tt MaxC} + 1$, then finish. 
	\\
	%
	%%%%
	%
	%\noindent
	{\bf\boldmath Step~4:} 
	Give ${a}_{\ell}$ jobs $J_{\ell,i} \hspace{1mm} (i = 0, 1, \ldots, {a}_{\ell}-1)$ 
	such that 
	$r(J_{\ell, i}) = s_{\ell} + p_{\ell} i/{a}_{\ell}$, 
	$d(J_{\ell, i}) = s_{\ell} + p_{\ell} (i+1)/{a}_{\ell}$, 
	in which 
	$p_{1} = 1$ 
	and 
	for any $\ell \geq 2$, $p_{\ell} = d(J_{\ell-1, 1}) - r(J_{\ell-1, 1})$. \\
	%
	%%%%
	%
	%\noindent
	{\bf\boldmath Step~5:} 
	Execute one of the following two cases: 
	\\
	%
	%%%%
	%
	%\noindent
	\hspace*{1mm}
	{\bf\boldmath Case~5.1 ($ON$ does not place any of $J_{\ell, i}$ onto the second machine):} 
	Finish. 
	\\
	%
	%%%%
	%
	%\noindent
	\hspace*{1mm}
	{\bf\boldmath Case~5.2 (Otherwise):} 
	Let $J_{\ell, i'}$ be the job which placed onto the second machine by $ON$ and whose release time is the closest to $s_{\ell}$. 
	$s_{\ell+1} :=$
	\hspace*{1mm}
	$s_{\ell} + p_{\ell} i'/{a}_{\ell}$. 
	Let ${a}_{\ell+1} > 0$ be an integer 
	such that 
	$(2 + \sum_{j = 1}^{\ell} \max \{ \frac{1}{x_1 + 1}, \frac{1}{x_2 + j} \} {a}_{j}) / {a}_{\ell+1}$ is sufficiently small. 
	%%%\com{%%ed:enough%sufficiently%}
	$\ell := \ell + 1$ and go to Step~3. 
	\\
	\noindent\vspace{-1mm}\rule{\textwidth}{0.5mm} 
	Let {\tt FinC} denote the value of $\ell$ at a time when the routine finishes. 
	%%%\com{%%%%%%}
	%
	%
	Let $\sigma_{{\tt FinC}}$ denote the input given to $ON$. 
	For a given input $\sigma_{{\tt FinC}}$, 
	an offline algorithm $OFF$ places ${a}_{{\tt FinC}}$ jobs $J_{{\tt FinC}, i}$ given at the time of the final execution, 
	that is, the ${\tt FinC}$th execution of Step~4, 
	onto the second machine, and 
	places the other jobs onto the first machine. 
	Thus, 
	for any ${\tt FinC} \leq {\tt MaxC}$, 
	\[
		V_{OPT}(\sigma_{{\tt FinC}}) \geq V_{OFF}(\sigma_{{\tt FinC}}) \geq {a}_{{\tt FinC}}
	\]
	and 
	\[
		V_{OPT}(\sigma_{{\tt MaxC}+1}) \geq V_{OFF}(\sigma_{{\tt MaxC}+1}) \geq {a}_{{\tt MaxC}}. 
	\]
	Conversely, 
	we consider the profit gained by $ON$. 
	The number of jobs which are placed by $ON$ on the $j$th machine for each $j \in \{ 1, 2 \}$ 
	immediately after Step~2 is $x_j$, and 
	$ON$ can gain profit $1/x_j$ per job. 
	%%%\com{%%ed:the profit%profit%}
	Since the profit of a job does not increase, 
	the total profit which $ON$ gains from all the jobs given in Step~2 is at most two at the end of the input. 
	%%%\com{%%ed:at%in%}
	%
	%
	
	%
	By the definition of Case~5.2, 
	jobs in the $\ell (\in [2, {\tt FinC}-1])$th execution of Step~4 are given in an interval 
	where $ON$ places jobs given in the $\ell-1$st execution of Step~4
	(Fig.~\ref{fig:glb2}). 
	Thus, 
	there exist $x_2 + \ell - 1$ jobs during the interval $[s_{\ell}, s_{\ell} + p_{\ell}]$ on $ON$'s second machine 
	immediately before the $\ell$th execution of Step~4. 
	Hence, 
	in the case where Case~5.2 is executed after the $\ell$th Step~4, 
	the profit of a job given in the $\ell$th Step~4 is at most $\max \{ \frac{1}{x_1 + 1}, \frac{1}{x_2 + \ell} \}$ 
	immediately after the $\ell$th execution of Step~4. 
	%%%\com{%%ed:just%immediately%}
	%
	Thus, 
	the total profit of the jobs given in the $\ell$th Step~4 is at most 
	$\max \{ \frac{1}{x_1 + 1}, \frac{1}{x_2 + \ell} \} {a}_{\ell}$. 
	Conversely, 
	if Case~5.1 is executed, 
	the profit of a job given in the $\ell$th Step~4 is at most $\frac{1}{x_1 + 1}$. 
	Thus, 
	the total profit is at most $\frac{{a}_{\ell}}{x_1 + 1}$. 
	Additionally, 
	$x_1 \geq {\tt MaxC}$ 
	because both $x_1 \geq x_2$ and $x_1 + x_2 = 2 \cdot {\tt MaxC}$ by definition. 
	Therefore, 
	if the routine finishes in Case~5.1 when ${\tt FinC} \leq {\tt MaxC}$, 
	we have 
	\begin{eqnarray*}
		V_{ON}(\sigma_{\tt FinC}) 
			&\leq& 2 + \sum_{j = 1}^{{\tt FinC}-1} \max \left \{ \frac{1}{x_1 + 1}, \frac{1}{x_2 + j} \right \} {a}_{j}
						+ \frac{{a}_{{\tt FinC}}}{x_1 + 1} \\
			&<& 2 + \sum_{j = 1}^{{\tt FinC}-1} \max \left \{ \frac{1}{x_1 + 1}, \frac{1}{x_2 + j} \right \} {a}_{j}
						+ \frac{{a}_{{\tt FinC}}}{{\tt MaxC}}, 
	\end{eqnarray*}
	where the second inequality follows from the fact that $x_1 \geq {\tt MaxC}$. 
	%%%\com{%%ed:fact%fact that%}
	%
	%
	In addition, 
	if the routine finishes in Step~3, 
	that is, ${\tt FinC} = {\tt MaxC}+1$,
	we have 
	\begin{eqnarray*}
		V_{ON}(\sigma_{{\tt MaxC}+1}) 
			&\leq& 2 + \sum_{j = 1}^{{\tt MaxC}} \max \left \{ \frac{1}{x_1 + 1}, \frac{1}{x_2 + j} \right \} {a}_{j}\\
			&\leq& 2 + \sum_{j = 1}^{{\tt MaxC}-1} \max \left \{ \frac{1}{x_1 + 1}, \frac{1}{x_2 + j} \right \} {a}_{j}
						+ \frac{{a}_{{\tt MaxC}}}{{\tt MaxC}}. 
	\end{eqnarray*}
	By the above argument, 
	if ${\tt FinC} \leq {\tt MaxC}$, 
	\begin{eqnarray*}
		\frac{ V_{OPT}(\sigma_{{\tt FinC}}) }{ V_{ON}(\sigma_{{\tt FinC}}) } 
			&\geq& \frac{ {a}_{{\tt FinC}} }{ 2 + \sum_{j = 1}^{{\tt FinC}-1} \max \{ \frac{1}{x_1 + 1}, \frac{1}{x_2 + j} \} {a}_{j}  +  {a}_{{\tt FinC}} / {\tt MaxC} }\\
			&\geq& \frac{ {a}_{{\tt FinC}} }{ {a}_{{\tt FinC}} \delta_{{\tt FinC}} +  {a}_{{\tt FinC}} / {\tt MaxC} }
			\geq \frac{ 1 }{ \delta_{{\tt FinC}} + 1 / {\tt MaxC} }, 
	\end{eqnarray*}
	where $\delta_{{\tt FinC}} = (2 + \sum_{j = 1}^{{\tt FinC}-1} \max \{ \frac{1}{x_1 + 1}, \frac{1}{x_2 + j} \}) / {a}_{{\tt FinC}}$. 
	Note that as both $\delta_{{\tt FinC}}$ and $1/{\tt MaxC}$ tend to zero, 
	$1/(\delta_{{\tt FinC}} + 1/{\tt MaxC})$ tends to infinity. 
	In a similar manner, 
	\[
		\frac{ V_{OPT}(\sigma_{{\tt MaxC}+1}) }{ V_{ON}(\sigma_{{\tt MaxC}+1}) } 
			\geq {\tt MaxC}. 
	\]
	\fi
\end{proof}
\ifnum \count12 > 0
\begin{figure*}[ht]
	 \begin{center}
	  \includegraphics[width=130mm]{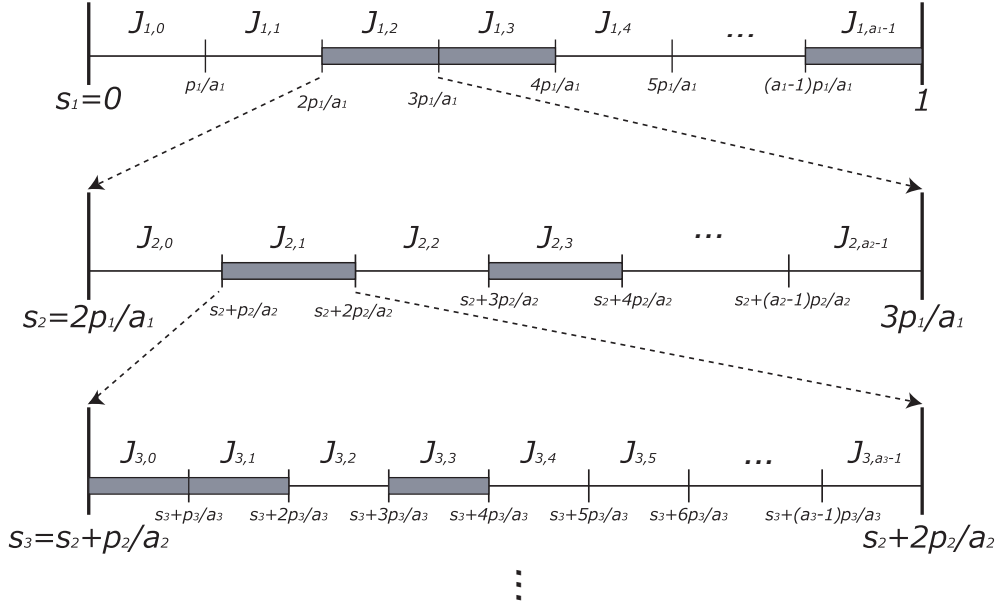}
	 \end{center}
	 \caption{
\ifnum \count10 > 0
%
%%%%\com{%%%%%%}
%
%%%%%Step~4%Case~5.1%%%%%
%
Step~4%%%%%%%%%%%%$J_{1, i} \hspace{1mm} (i = 1, \ldots, a_{1}-1)$%%%%%
$ON$%$J_{1, 0}$%$J_{1, 1}$%1st%%%%place%%
$J_{1, 2}$%2nd%%%%place%%%%%%%%%%%%%%
%%%%%%%$ON$%2nd%%%%place%%%%%%%%%
%
%%%%%
Step~4%2%%%%%%%%%%
$s_{2} = \frac{2 p_{1}}{a_{1}}$%%%%%
%%$[r(J_{1, 2}), d(J_{1, 2})] = [s_{2}, s_{2} + p_{2}]$%%
%%%$J_{2, i} \hspace{1mm} (i = 0, \ldots, a_{2}-1)$%%%%%%%
%
%%%%
%%$[r(J_{2, 1}), d(J_{2, 1})]$%%
%%%$J_{3, i} \hspace{1mm} (i = 0, \ldots, a_{3}-1)$%%%%%%%
%
%
%%%%
$\ell$%%%Step~4%%%%%%%%
$ON$%2nd%%%%%
%%$[s_{\ell}, s_{\ell} + p_{\ell}]$
%%$x_2 + \ell - 1$%%job%%%%%%%%%%%%%%%%%%%%
%
%
\fi
\ifnum \count11 > 0
%
%%%\com{%%%%%}
%
Execution example of Step~4 and Case~5.1 in the routine. 
For jobs $J_{1, i} \hspace{1mm} (i = 0, \ldots, a_{1}-1)$ given in Step~4, 
the top figure shows that 
$ON$ places $J_{1, 0}$ and $J_{1, 1}$ onto the first machine and 
$J_{1, 2}$, 
which is represented by a gray rectangle, 
onto the second machine. 
At the second execution of Step~4, 
$s_{2} = \frac{2 p_{1}}{a_{1}}$ and 
jobs $J_{2, i} \hspace{1mm} (i = 0, \ldots, a_{2}-1)$ are given 
in the interval $[r(J_{1, 2}), d(J_{1, 2})] = [s_{2}, s_{2} + p_{2}]$. 
In the same way, 
jobs $J_{3, i} \hspace{1mm} (i = 0, \ldots, a_{3}-1)$ are given 
in $[r(J_{2, 1}), d(J_{2, 1})]$. 
Thus, 
$x_2 + \ell - 1$ jobs are placed on $ON$'s second machine in $[s_{\ell}, s_{\ell} + p_{\ell}]$ 
immediately after the $\ell$th execution of Step~4. 
\fi
			}
	\label{fig:glb2}
\end{figure*}
\fi
\ifnum \count11 > 0
%
%%%%\com{%%%%%%}
%

%

%
\fi
\ifnum \count11 > 0
%
%%%%\com{%%%%%}
%
%
\fi
\begin{theorem}\label{thm:GLBm}
	\ifnum \count10 > 0
	%
	%%%\com{%%%%%%}
	%
	%%%$m$%%%%
	%%%%%%%%%%%%%%%%%%%
	%%%%%%%
	%
	
	%
	\fi
	\ifnum \count11 > 0
	%
	%%%\com{%%%%%}
	%
	For any $m$, 
	there does not exist a competitive deterministic algorithm. 
	\fi
\end{theorem}
\begin{proofsketch}
	\ifnum \count10 > 0
	%
	%%%\com{%%%%%%}
	%
	%%\ref{thm:GLB2}%%%
	%%%%%%%%%%%$ON$%%%%%%%%%%%%%%%%%%%%job%%%
	%%%%%%%%%%%${\tt MaxC}$%%%%%%%%%%%%%%
	%
	%%%%%%%
	Step~4%%%%%%%%%${a}_{{\tt MaxC}}$ (${a}_{{\tt FinC}}$)%%job%%%%%%%%%%
	$OPT$%${a}_{{\tt MaxC}}$ (${a}_{{\tt FinC}}$)%%%%%%%%%%%%%%%%
	$ON$%%%%%%$\frac{ {a}_{{\tt MaxC}} }{ {\tt MaxC} }$ ($\frac{ {a}_{{\tt FinC}} }{ {\tt MaxC} }$)%%%%%%%%%%%%%%
	%
	%%%%%%%%%%%
	$ON$%$m$%%%%%%%%%%%%%%%%%%%job%%%
	${\tt MaxC}$%%%%%%%%%%%$\sigma$%%%%%%%%%%%%
	%
	%%%%%%
	$\frac{ V_{OPT}(\sigma) }{ V_{ON}(\sigma) } \geq {\tt MaxC}$
	%%%%%%
	%
	
	%
	\fi
	\ifnum \count11 > 0
	%
	%%%\com{%%%%%}
	%
	In Theorem~\ref{thm:GLB2}, 
	we use the routine to force the numbers of jobs existing during an interval on both machines of an online algorithm $ON$ to approach a sufficiently large integer ${\tt MaxC}$. 
	%%%\com{%%ed:enough%sufficiently%}
	%
	Thus, 
	while $OPT$ can gain the profit of ${a}_{{\tt MaxC}}$ (${a}_{{\tt FinC}}$) from ${a}_{{\tt MaxC}}$ (${a}_{{\tt FinC}}$) jobs given in Step~4 of the routine, 
	$ON$ gains only at most approximately $\frac{ {a}_{{\tt MaxC}} }{ {\tt MaxC} }$ ($\frac{ {a}_{{\tt FinC}} }{ {\tt MaxC} }$). 
	%%%\com{%%ed:apx at most%at most apx%%%%%%%%%%%%%%%%}
	%
	In a similar manner, 
	we can construct the input $\sigma$ to force the number of jobs which $ON$ places during an interval on each of $m$ machines to approach ${\tt MaxC}$. 
	As a result, 
	we have $\frac{ V_{OPT}(\sigma) }{ V_{ON}(\sigma) } \geq {\tt MaxC}$. 
	\fi
\end{proofsketch}
%

%%%%%%%%%%%%%%%%%%%%%%%%%%%%%%%%%%%%%%%%%%%%%%%%%
\section{Upper Bounds for Uniform Profit Case} \label{sec:ubuc}
\ifnum \count10 > 0
%
%%%\com{%%%%%%}
%
%
%%%%%
%%%%%%%%%%%%%%%%%%%%%uniform profit case%%%%%
%
%
%%%%%%%
1%%%%%%%%%%%%%%%%%%%%%%%%%%%%%
%
%%%%
%%%%%%%%%%%%%%%%%%%%%%%%%%%%
%

%
\fi
\ifnum \count11 > 0
%
%%%\com{%%%%%}
%
In this section, 
we consider the uniform profit case, that is, 
the case in which the profit of a job is equal to its length. 
In this case, 
the total amount of time during which at least one job is scheduled on a machine is equal to the total amount of the satisfaction gained on the machine. 
That is, 
the objective of this case can be regarded as maximizing the working hours of all the machines.  
\fi
%

%%%%%%%%%%%%%%%%%%%%%%%%%%%%%%%%%%%%%%%%%%%%%%%%%%%%%%
\subsection{Preliminaries}
\ifnum \count10 > 0
%
%%%\com{%%%%%%}
%
2%%%%$I = [t_{1}, t_{2}]$%$I' = [t'_{1}, t'_{2}]$%%%%%
$t'_{1} < t_{2}$%%$t_{1} < t'_{2}$%%%%%%%%
$I$%$I'$%{\em intersect}%%%%%%
%
%
%%%$J$%%%%%
%%$[ r(J), d(J) ]$%{\em $J$%%%}%%%%
%
%
%%%%%%$ALG$%2%%%%%%%%%%%%place%%
%%%%%%%intersect%%%%%
%%%%{\em overlap}%%%%%%
%
%
%%%%%$I = [t, t']$%%%%%
$t' - t$%{\em $I$%%%}%%%%
$|I|$%%%%
%

%
%
%%%%%%%%%%%%%%%%$GR$
%%%%%%%%
%
$GR$%%%%$J$%%%%%%%%%
%%%%%%%%%%%%%%%%$J$%%%%%%
%
%%%%%%%%%%%%%%%%%%%%
%

%
%%%%%%%%%%%%%%%%
%%%%%%%%%%%
%
%
%%%%%%%%%%%%%%%%%%$J_1$%$J_2$%%%%%
%%$I$%%%%%%%2%%%%%%overlap%%%%%%%%
%%%$J_{1}$%%%%%%%%%%%%%%%%%%%%%%%%%%%%
%%%%%
%%$I$%%%$J_{1}$%$J_{2}$%%%%%
%%%%$|I|$%0%%%%%%%%
%
%%%%%
%%%%%%%%%%%%%%%%%%%%%%%%%%%%
%%%%%%%%%%%%%%%%%%%%%%%%%
%
%
uniform profit case%%%
%%%%%%%%%%%place%%%%%%%%%%%%%%%%%%%%%
%%%%%place%%%%%%%%%1%%%%%%%%%%%%%%%%%%%
%%%%%%%%%%%%%%%%%%%%%%%
%
%
%%%%%
%%%%%%%%%%%%%%%%%%%%%%%%%%%%%%%
%

%
\fi
\ifnum \count11 > 0
%
%%%\com{%%%%%}
%
After the end of the input, 
we need to evaluate the profit from each job by $OPT$ using the profits yielded from intervals of jobs scheduled by $GR$ to analyze the performance of $GR$. 
Then, 
we classify intervals (or points) in a job $J$ by $GR$ or $OPT$ into the following four categories 
depending on the behaviors of $GR$ and $OPT$ for $J$. 
For any two intervals $I = [t_{1}, t_{2}]$ and $I' = [t'_{1}, t'_{2}]$, 
we say that $I$ {\em intersects} with $I'$ 
if $t'_{1} < t_{2}$ and $t_{1} < t'_{2}$. 
For any job $J$, 
we call the interval $[ r(J), d(J) ]$ the {\em interval of $J$}. 
If an algorithm $ALG$ places two jobs onto the same machine and they intersect, 
then we say that they {\em overlap}. 
For any interval $I = [t, t']$, 
we call the value of $t' - t$ the {\em length of $I$}, 
written as $|I|$. 
We give the definition of a greedy algorithm $GR$ and 
analyze its performance in this section. 
$GR$ places a given job $J$ onto the machine on which 
$GR$ gains the largest profit from $J$. 
%%%\com{%%%%%%%}
%
The tie-breaking rule selects the minimum indexed machine.
For ease of analyzing, 
we introduce the following idea. 
%%%\com{%%idea%%}
%
%
Suppose that 
two jobs $J_1$ and $J_2$ are placed onto the same machine, and 
they overlap in an interval $I$. 
Also, 
suppose that $J_1$ is the first job placed in $I$ on the machine.
Then, 
pretend that the profits from $I$ of $J_{1}$ and $J_{2}$ are $|I|$ and zero, respectively. 
That is, 
we pretend that 
a job which is placed chronologically first in an interval on a machine monopolizes the machine power in the interval. 
Note that 
in the uniform profit case, 
the total profit gained from an interval of jobs placed on a machine depends 
not on how large the number of the jobs in the interval is 
but on whether there exists at least one job placed in the interval. 
That is why this assumption does not affect the profit of any algorithm. 
\fi
%

%%%%%%%%%%%%%%%%%%%%%%%%%%%%%%%%%%%%%%%%%%%%%%%%%%%%%%
\subsection{Overview of Analysis}

\ifnum \count10 > 0
%
%%%\com{%%%%%%}
%
$GR$%%%%%%%%%%%%%%%%%%%%
$GR$%%%%%%%%%%%%%%%%$OPT$%%%%%%%%%%%%%%%%
%
%%%%
$GR$%%%%%$OPT$%%%%%%%%%%%%%%%%4%%%%%%%
%

%
%%%%%%$J$%%%$I \subseteq [ r(J), d(J) ]$%%%%%
$GR$%$J$%$I$%%%%%%%0%%%%
$OPT$%%%%$|I|$%%%%%%
$J$%$I$%$J$%{\em $OPT$ extra interval}%%%
%%%%%%{\em $oe$-interval}%%%%%
%
%%~\ref{fig:ivcat}%%%
%
%
%%%
$OPT$%$J$%$I$%%%%%%%0%%%%
$GR$%%%%$|I|$%%%%%%
$J$%$I$%$J$%{\em $GR$ extra interval}%%%
%%%%%%{\em $ge$-interval}%%%%%
%
%
$GR$%$OPT$%%%$J$%$I$%%%%%%%$|I|$%%%%%%
$J$%$I$%$J$%{\em common interval}%%%
%%%%%%{\em $c$-interval}%%%%%
%
%
%%%%%%%%%
$c$-interval%$ge$-interval%{\em profit interval}%%%
%%%%%%{\em $p$-interval}%%%%%
%
%
$GR$%$OPT$%%%$J$%$I$%%%%%%%$0$%%%%%%
$J$%$I$%$J$%{\em non-profit interval}%%%
%%%%%%{\em $n$-interval}%%%%%
%
%
%%%
$J$%$oe$-interval%$ge$-interval, $c$-interval, $p$-interval%
%%%%%%%%%$J$%%%%%
$J$%{\em $oe$-fraction}%{\em $ge$-fraction}, {\em $c$-fraction}, {\em $p$-fraction}%
%%%%
%

%
%%%%%%%%%%%%%%%
%$oe$-fraction%%%%%%$oe$-interval%%$p$-fraction%%%%%%$p$-interval%%
%%%%%%%%%%%%$GR$%%%%%%%%%%
%
%
%%%%%%%%%%
%%%%%%%%%%%%
%
%
$c$-interval%%%%%%$oe$-interval%%%%$V_{oe}(\sigma)$%%%%
$ge$-interval%%%%%%$oe$-interval%%%%$V_{oe'}(\sigma)$%%%%
%
%%%
$c$-interval%%%%$V_{c}(\sigma)$, 
$ge$-interval%%%%$V_{ge}(\sigma)$%%%%%%%
%
%
%%%%%%%%%%
\begin{equation} \label{sec:upanym.eq.1}
	V_{GR}(\sigma) = V_{c}(\sigma) + V_{ge}(\sigma)
\end{equation}
\begin{equation} \label{sec:upanym.eq.2}	
	V_{OPT}(\sigma) = V_{c}(\sigma) + V_{oe}(\sigma) + V_{oe'}(\sigma)
\end{equation}
%%%%%%
%
%
%%%%%%%%%%%%%%%%%%%%3%%%%%%%%%%%%%%%
%
\begin{description}
\itemsep=-2.0pt
\setlength{\leftskip}{10pt}
\item[1.]
%%%$oe$-fraction%%%%%$p$-fraction%1%%%%%%%
%
\item[2.]
$GR$%%%%%%%%%%
%%%$c$-fraction%%%2%%%%%%%%%
%
\item[3.]
%%%$ge$-fraction%%%3%%%%%%%%%
%
\end{description}
%
%
%%%%%%%%%%
$oe$-interval%%$p$-interval%%3%%%%$M_{1},M_{2},M_{3}$%%%%%%%%%%%%%%%%%%
%
%%%
%%%$oe$-fraction $f$%$f'(\ne f)$%%%%$i \in \{ 1, 2, 3 \}$%%%%
$M_{i}(f) = \varnothing$%%%%%%
%
%%%%%%
%%%$oe$-fraction $f$%%%%%
$M_{1}(f) \cup M_{2}(f) \cup M_{3}(f) \ne \varnothing$
%%%%%%
%
$M_{1}(f) \ne \varnothing$%%%%
$p$-fraction $f'$%%%%%%
$M_{1}(f) = f'$%%%%%%
$M_{2}(f) \ne \varnothing$%%%%
$ge$-fraction $f'$%%%%%%
$M_{2}(f) = f'$%%%%%%
$M_{3}(f) \ne \varnothing$%%%%
$p$-fraction $f'$%%%%%%
$M_{3}(f) = f'$%%%%%%
%
%
%%%$oe$-fraction $f$%$f'(\ne f)$%%%%$i \in \{ 1, 2, 3 \}$%%%%%
$M_{i}(f) \cap M_{i}(f')$
%%%%%%
%
%
%%%%%%
%%%%%%%%%%
\begin{equation} \label{sec:upanym.eq.3}
	V_{oe}(\sigma) \leq 2 V_{c}(\sigma)
\end{equation}
\begin{equation} \label{sec:upanym.eq.4}	
	V_{oe'}(\sigma) \leq 3 V_{ge}(\sigma)
\end{equation}
%%%%%%
%
%
%(\ref{sec:upanym.eq.2})%%%
\begin{eqnarray*}
	V_{OPT}(\sigma) 
		&=& V_{c}(\sigma) + V_{oe}(\sigma) + V_{oe'}(\sigma) \\
		&\leq& V_{c}(\sigma) + 2 V_{c}(\sigma) + 3 V_{ge}(\sigma) \mbox{\hspace{6mm} (%~(\ref{sec:upanym.eq.3})%(\ref{sec:upanym.eq.4})%%)} \\
		&=& 3 (V_{c}(\sigma) + V_{ge}(\sigma))
		= 3 V_{GR}(\sigma) \mbox{\hspace{6mm} (%~(\ref{sec:upanym.eq.1})%%)} 
\end{eqnarray*}
%%%%%%
%
%
%%%%%%%%%%%%
%

%
\fi
\ifnum \count11 > 0
%
%%%\com{%%%%%}
%
To evaluate the performance of $GR$, that is, its competitive ratio, 
we bound the profit of $OPT$ at the end of the input using that of $GR$. 
Then, 
we classify intervals of jobs placed by either $GR$ or $OPT$ into four categories. 
For any job $J$ and any interval $I \subseteq [ r(J), d(J) ]$, 
if the profit gained from $I$ of $J$ by $GR$ is zero and 
that by $OPT$ is $|I|$, 
then we call $I$ of $J$ an {\em $OPT$ extra interval of $J$}
(denoted as an {\em $oe$-interval}, for short)
(see Fig.~\ref{fig:ivcat}). 
Also, 
if the profit gained from $I$ of $J$ by $OPT$ is zero and 
that by $GR$ is $|I|$, 
then we call $I$ of $J$ a {\em $GR$ extra interval of $J$}
(a {\em $ge$-interval}, for short).
If the profits gained from $I$ of $J$ by $GR$ and $OPT$ are both $|I|$, 
we call $I$ of $J$ a {\em common interval of $J$}
(a {\em $c$-interval}, for short).
For ease of presentation, 
we call an interval which is a $c$-interval or a $ge$-interval a {\em profit interval}
(a {\em $p$-interval}, for short). 
If the profits gained from $I$ of $J$ by $GR$ and $OPT$ are both zero, 
we call $I$ of $J$ a {\em non-profit interval of $J$}
(an {\em $n$-interval}, for short).
Further, 
we call a point in an $oe$-interval (a $ge$-interval, a $c$-interval, and a $p$-interval, respectively) of $J$ 
an {\em $oe$-fraction} (a {\em $ge$-fraction}, a {\em $c$-fraction}, and a {\em $p$-fraction}, respectively) of $J$. 
We evaluate the competitive ratio of $GR$ by ``assigning'' $p$-fractions (i.e., $p$-intervals) to all $oe$-fractions (i.e., $oe$-intervals) according to a routine, which is defined later.
This ``assignment'' is realized by some functions. 
Let $V_{oe}(\sigma)$ be the total length of $oe$-intervals to which $c$-intervals are assigned. 
Let $V_{oe'}(\sigma)$ be the total length of $oe$-intervals to which $ge$-intervals are assigned. 
Also, 
let $V_{c}(\sigma)$ be the total length of $c$-intervals and 
$V_{ge}(\sigma)$ be the total length of $ge$-intervals. 
Then, we have by definition, 
\begin{equation} \label{sec:upanym.eq.1}
	V_{GR}(\sigma) = V_{c}(\sigma) + V_{ge}(\sigma)
\end{equation}
and 
\begin{equation} \label{sec:upanym.eq.2}	
	V_{OPT}(\sigma) = V_{c}(\sigma) + V_{oe}(\sigma) + V_{oe'}(\sigma). 
\end{equation}
We will show the following three properties of the assignments by the routine: 
\begin{description}
\itemsep=-2.0pt
\setlength{\leftskip}{10pt}
\item[1.]
Each $oe$-fraction is assigned a $p$-fraction, 
\item[2.]
a $c$-fraction of a job given to $GR$ is assigned at most twice, and 
\item[3.]
a $ge$-fraction is assigned at most three times. 
\end{description}
To show these, 
we will construct sequentially three functions $M_{1},M_{2}$ and $M_{3}$ from $oe$-intervals to $p$-intervals satisfying the following properties: 
Initially, 
for any $oe$-fraction $f$ and any $i \in \{ 1, 2, 3 \}$, 
$M_{i}(f) = \varnothing$. 
At the end of the input, 
for any $oe$-fraction $f$, 
$M_{1}(f) \cup M_{2}(f) \cup M_{3}(f) \ne \varnothing$. 
There exists a $p$-fraction $f'$ such that $M_{1}(f) = f'$ 
if $M_{1}(f) \ne \varnothing$. 
There exists a $ge$-fraction $f'$ such that $M_{2}(f) = f'$ 
if $M_{2}(f) \ne \varnothing$.  
There exists a $p$-fraction $f'$ such that $M_{3}(f) = f'$ 
if $M_{3}(f) \ne \varnothing$. 
For any $oe$-fractions $f$ and $f'(\ne f)$ and any $i \in \{ 1, 2, 3 \}$, 
$M_{i}(f) \cap M_{i}(f') = \varnothing$. 
Then, we have by these functions, 
\begin{equation} \label{sec:upanym.eq.3}
	V_{oe}(\sigma) \leq 2 V_{c}(\sigma)
\end{equation}
and 
\begin{equation} \label{sec:upanym.eq.4}
	V_{oe'}(\sigma) \leq 3 V_{ge}(\sigma). 
\end{equation}
By Eq.~(\ref{sec:upanym.eq.2}), we have 
\begin{eqnarray*}
	V_{OPT}(\sigma) 
		&=& V_{c}(\sigma) + V_{oe}(\sigma) + V_{oe'}(\sigma) \\
		&\leq& V_{c}(\sigma) + 2 V_{c}(\sigma) + 3 V_{ge}(\sigma)\mbox{\hspace{6mm} (by Eqs.~(\ref{sec:upanym.eq.3}) and (\ref{sec:upanym.eq.4}))}\\
		&=& 3 (V_{c}(\sigma) + V_{ge}(\sigma)) 
		= 3 V_{GR}(\sigma), \mbox{\hspace{6mm} (by Eq.~(\ref{sec:upanym.eq.1}))} \\
\end{eqnarray*}
which leads to the following theorem: 
\fi
\ifnum \count12 > 0
\begin{figure*}[ht]
	 \begin{center}
	  \includegraphics[width=100mm]{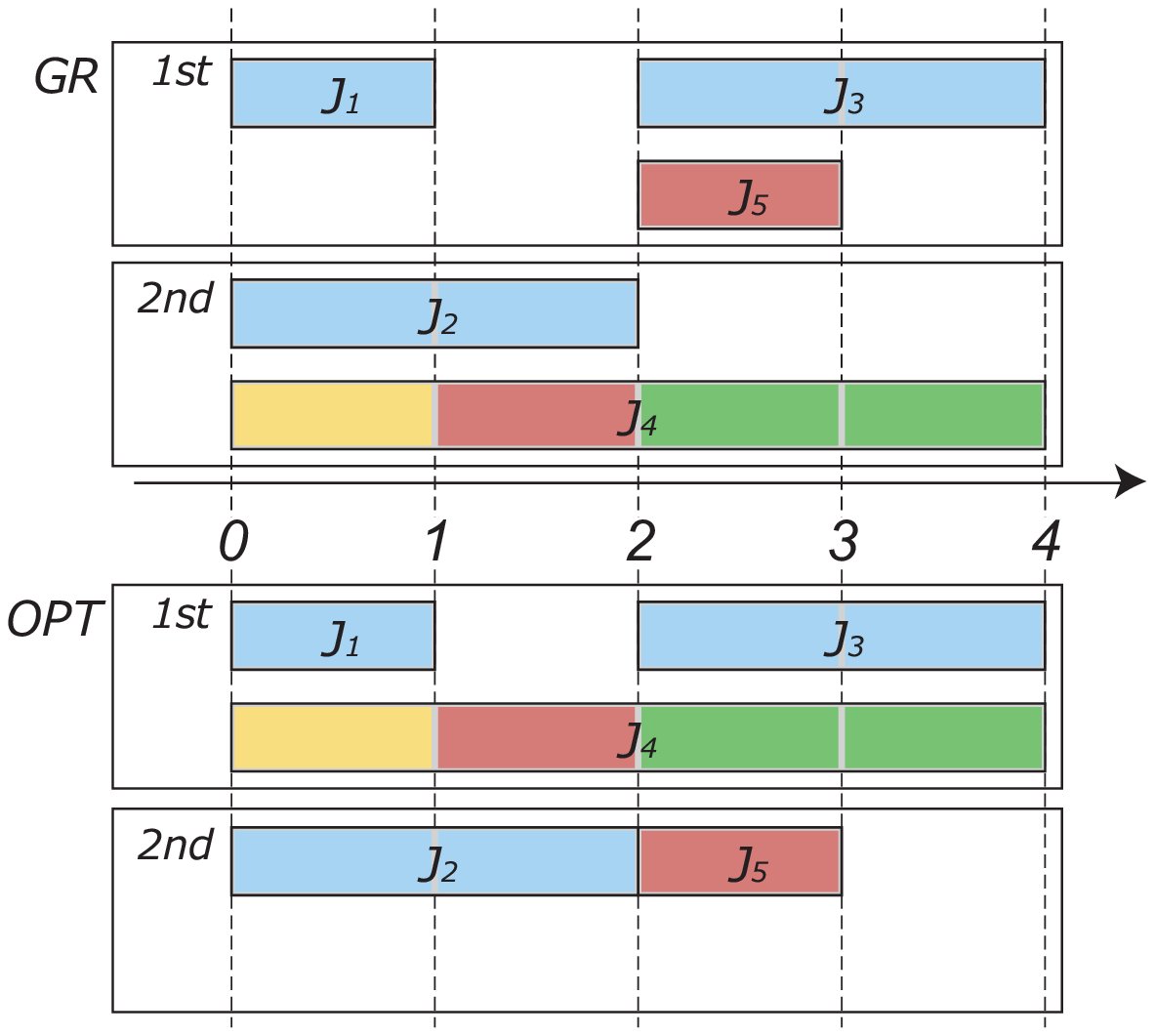}
	 \end{center}
	 \caption{
\ifnum \count10 > 0
%
%%%\com{%%%%%%}
%
2%%%%%%5%%%%%$J_{i} \hspace{1mm} (i \in [1, 5])$%%%%%%%%%%%
%
%%%$GR$%%%%%%%%%%%
%%%$OPT$%%%%%%%%%%%
%
$c$-interval%%%%%%%%
$ge$-interval%%%%%%%%
$oe$-interval%%%%%%%%
$n$-interval%%%%%%%%%%%%%%
%
%%%%
$J_{4}$%%%$[1, 2]$%%%
$GR$%$J_{2}$%overlap%%%%%
%%%
$OPT$%overlap%%%%%%
%%%%%%%%%$oe$-interval%%%%
%
%%%
$J_{4}$%%%$[0, 1]$%%%
$GR$%$OPT$%%%%%%%%%%%%
$n$-interval%%%%
\fi
\ifnum \count11 > 0
%
%%%\com{%%%%%}
%
An example in which five jobs $J_{i} \hspace{1mm} (i \in [1, 5])$ are given when $m=2$. 
The top (bottom) figure presents jobs placed by $GR$ ($OPT$). 
$c$-intervals ($ge$-intervals, $oe$-intervals and $n$-interval, respectively)
are shown in blue (green, red and yellow, respectively) squares. 
The intervals $[1, 2]$ of $J_{4}$ and $J_{2}$ by $GR$ overlap 
but the intervals $[1, 2]$ of $J_{4}$ by $OPT$ do not. 
Thus, 
this interval of $J_{4}$ is an $oe$-interval. 
Also, 
neither $GR$ nor $OPT$ gains the profit from the interval $[0, 1]$ of $J_{4}$ and thus 
this interval is an $n$-interval. 
\fi
			}
	\label{fig:ivcat}
\end{figure*}
\fi
\begin{theorem}\label{thm:upanym}
	\ifnum \count10 > 0
	%
	%%%\com{%%%%%%}
	%
	%%%$m \geq 2$%%%%%
	$GR$%%%%%%%3%%%%
	\fi
	\ifnum \count11 > 0
	%
	%%%\com{%%%%%}
	%
	For any $m \geq 2$, 
	the competitive ratio of $GR$ is at most three. 
	\fi
\end{theorem}
%

%%%%%%%%%%%%%%%%%%%%%%%%%%%%%%%%%%%%%%%%%%%%%%%%%%%%%%
\subsection{Analysis of $GR$}
\ifnum \count10 > 0
%
%%%\com{%%%%%%}
%
%%%%%%$J$%%%%%$t \in [ r(J), d(J) ]$%%%%%
%%$[ r(J), t ]$%%%%%%$J$%$oe$-interval%%%%$E(J, t)$%%%%
%
%
%%%%%%$J$%$J$%%%%%%%%%%%%%%%$J'$%
%%$[ t_{1}, t_{2} ]$%%%%$a (\in [1, m])$%%%%%
%%%$J$%%%%%%%%%%%%%
$GR$%$a$%%%%%%%%%$[ t_{1}, t_{2} ]$%%%%%%
%%%%%$J'$%$n$-interval%%%%%%%%%
%%$J'$%$n$-interval%intersect%%%%%%
$p$-interval%%%%$P_{a}(J, J', t_{1}, t_{2})$%%%%
%
%
%%%$a (\in [1, m])$%%%%%%%$J$%%
$J$%%%%%%%%%%%%%%%$J'$%
%%%%$t \in [ r(J'), d(J') ]$%%%%%
%%%$J$%%%%%%%%%%%%%%%%%
$h_{a}(J, J', t) = t'$%%%%%%
%%%%$P_{a}( J, J', r(J'), t' ) = E(J', t)$
%%
$t' \in [ r(J'), d(J') ]$
%%%%%%
%
%%%%%%%~\ref{LMA:peqx}%%%
$t'$%%%%%%%%%
%
%
%%%
%%%$i \in \{ 1, 2, 3 \}$%%%%$p$-fraction $f'$%%%%%
$M^{-1}_{i}(f') = \{ f \mid M_{i}(f) = f' \}$
%%%%%%
%
%
$M_{1}^{-1}(f') = \varnothing$%%%%%%%%$c$-fraction $f'$%%
{\em 1-assignable}%%%%%%%
$M_{2}^{-1}(f') = \varnothing$%%%%%%%%$ge$-fraction $f'$%%
{\em 2-assignable}%%%%%%%
$M_{2}^{-1}(f') \ne \varnothing$%%$M_{1}^{-1}(f') = \varnothing$%%%%%%%%$ge$-fraction $f'$%%
{\em 1-assignable}%%%%%%%
1-assignable%%%%%2-assignable%%%$p$-fraction%{\em assignable}%%%%%%%
%
%
%%%3%%%%%%%%%%%%%%%%%%%%%%%%%%%%%
%
%
%%%
%%%%%%%%%%%%%%%
Appendix~\ref{ap.sec:assex}%%%%%%%%
\fi
\ifnum \count11 > 0
%
%%%\com{%%%%%}
%
For any job $J$ and any point $t \in [ r(J), d(J) ]$, 
let $E(J, t)$ denote the total length of $oe$-intervals of $J$ in the interval $[ r(J), t ]$. 
For any job $J$, any job $J'$ given before $J$, 
any interval $[ t_{1}, t_{2} ]$ and any $a (\in [1, m])$, 
let $P_{a}(J, J', t_{1}, t_{2})$ denote the total length of $p$-intervals of $GR$'s jobs placed on the $a$th machine 
which are in $[ t_{1}, t_{2} ]$ immediately after $J$ is placed and are not intersecting with any $n$-interval of $J'$. 
For any $a (\in [1, m])$, 
any job $J$, any job $J'$ given before $J$, and 
any point $t \in [ r(J'), d(J') ]$, 
define $h_{a}(J, J', t) = t'$ 
in which $t'$ is the point such that 
$P_{a}( J, J', r(J'), t' ) = E(J', t)$ 
and $t' \in [ r(J'), d(J') ]$ 
immediately after $J$ is placed onto the machine. 
($t'$ exists by Lemma~\ref{LMA:peqx}, 
which is shown later.)
For any $i \in \{ 1, 2, 3 \}$ and any $p$-fraction $f'$, 
define $M^{-1}_{i}(f') = \{ f \mid M_{i}(f) = f' \}$. 
We say that a $c$-fraction $f'$ such that $M_{1}^{-1}(f') = \varnothing$ is {\em 1-assignable}. 
We say that a $ge$-fraction $f'$ such that $M_{2}^{-1}(f') = \varnothing$ is {\em 2-assignable}. 
We say that a $ge$-fraction $f'$ such that $M_{2}^{-1}(f') \ne \varnothing$ and $M_{1}^{-1}(f') = \varnothing$ is {\em 1-assignable}. 
If a $p$-fraction is 1-assignable or 2-assignable, 
we say that it is {\em assignable}. 
Now we give the definition of the routine mentioned in the previous section. 
For better understanding assignments, 
we give examples in Appendix~\ref{ap.sec:assex}. 
\fi
\ifnum \count10 > 0
%
%%%\com{%%%%%%\\}
%
\noindent\vspace{-1mm}\rule{\textwidth}{0.5mm} 
\vspace{-3mm}
{{\sc AssignmentRoutine}}\\
\rule{\textwidth}{0.1mm}
	$j$%%%%%%$J_{j}$%place%%%%%%%%%%
	${\cal J} := ($$J_{j}$%%
	$J_{j}$%%%%intersect%%%%%%%%%%%$J_{j'} \hspace{1mm} (j' \leq j-1)$%%%%%%$)$%
	%$J \in {\cal J}$%%$oe$-fraction $f$%%%%%%%%%%%%
	\\
	{\bf\boldmath Step~1:}
	%$i \in \{ 1, 2, 3 \}$%%%%%
	$M_{i}(f) := \varnothing$%
	$t_{1} := h_{1}(J_{j}, {J}, t)$%
	%%%%$f$%%$t$%%%%%%%%%\\
	%
%
	{\bf\boldmath Step~2:}
		%%2%%Case%%%%%%%%%\\
%
	\hspace*{1mm}
	{\bf\boldmath Case~2.1 ($t_{1}$%assignable%$p$-fraction $f_{1}$%%%%%%%):} 
		$f_{1}$%1-assignable%%%%%%
		$M_{1}(f) := f_{1}$%
		%
		%%%%%%$f_{1}$%2-assignable%%%%%%
		$M_{2}(f) := f_{1}$%
		\\
	\hspace*{1mm}
	{\bf\boldmath Case~2.2 ($t_{1}$%assignable%$p$-fraction%%%%%%%%):}
		%
		%%~\ref{LMA:ass}%%%
		$GR$%%%%%%%%%%%%$a (\in \{ 1, m \})$th%%%%%$t_{a}$%%
		$M^{-1}_{3}(f_{a}) = \varnothing$%%%%%%%$p$-fraction $f_{a}$%%%%%%
		%%%%
		$t_{a} = h_{a}(J_{j}, {J}, t)$%%%%%%
		$M_{3}(f) := f_{a}$%
		%%%~\ref{LMA:peqx}%%%%%%$a' \in \{ 1, m \}$%%%%%
		$t_{a'}$%%%%%%%
		\\
\noindent\vspace{-1mm}\rule{\textwidth}{0.5mm} 
\fi
\ifnum \count11 > 0
%
%%%\com{%%%%%\\}
%
\noindent\vspace{-1mm}\rule{\textwidth}{0.5mm} 
\vspace{-3mm}
{{\sc AssignmentRoutine}}\\
\rule{\textwidth}{0.1mm}
	Consider a moment immediately after the $j$th job $J_{j}$ is placed. 
	${\cal J} := ($the set of $J_{j}$ plus
	each job $J_{j'} \hspace{1mm} (j' \leq j-1)$ whose interval intersects with the interval of $J_{j} )$. 
	For any $oe$-fraction $f$ of each ${J} \in {\cal J}$, 
	execute the following. 
	\\
	{\bf\boldmath Step~1:}
	For each $i \in \{ 1, 2, 3 \}$, 
	$M_{i}(f) := \varnothing$. 
	$t_{1} := h_{1}(J_{j}, {J}, t)$, 
	in which $f$ exists at a point $t$. \\
	{\bf\boldmath Step~2:}
		Execute one of the following two cases. \\
	\hspace*{1mm}
	{\bf\boldmath Case~2.1 (An assignable $p$-fraction $f_{1}$ exists at $t_{1}$):} 
		If $f_{1}$ is 1-assignable, 
		$M_{1}(f) := f_{1}$. 
		Otherwise, 
		if $f_{1}$ is 2-assignable, 
		$M_{2}(f) := f_{1}$. 
		\\
	\hspace*{1mm}
	{\bf\boldmath Case~2.2 (No assignable $p$-fraction exists at $t_{1}$):}
		By Lemma~\ref{LMA:ass}, 
		there exists a $p$-fraction $f_{a}$ at the point $t_{a}$ 
		on some $a (\in \{ 1, m \})$th machine 
		such that $M^{-1}_{3}(f_{a}) = \varnothing$, 
		in which 
		$t_{a} = h_{a}(J_{j}, {J}, t)$. 
		(For any $a' \in \{ 1, m \}$, 
		there exists $t_{a'}$ by Lemma~\ref{LMA:peqx}.) 
		\\
\noindent\vspace{-1mm}\rule{\textwidth}{0.5mm} 
\fi
\ifnum \count10 > 0
%
%%%\com{%%%%%%}
%
%%%%%%%Case~2.2%$t_{a}$%%%%%%%%%%%
%
%%%%
Case~2.2%$p_{a}$%%%%%%%%%%%
%%%%%
%$oe$-fraction%%%%%
$p$-fraction%%%%%%%%%%%%%%
\fi
\ifnum \count11 > 0
%
%%%\com{%%%%%}
%
In the following, 
we first show the existence of $t_{a}$ in Case~2.2. 
Next, 
we show that there exists $p_{a}$ in Case~2.2. 
That is, 
we prove that the routine can assign a $p$-fraction to each $oe$-fraction. 
\fi
\begin{LMA} \label{LMA:peqx}
	\ifnum \count10 > 0
	%
	%%%\com{%%%%%%}
	%
	%%%$a (\in [1, m])$%%%%%%%$J$%%
	$J$%%%%%%%%%%%%%%%$J'$%
	%%%%$t \in [ r(J'), d(J') ]$%%%%%
	%%%$J$%%%%%%%%%%%%%
	$h_{a}(J, J', t) = t'$
	$t' \in [ r(J'), d(J') ]$
	%%%%%%$t'$%%%%%%
	%
	
	%
	\fi
	\ifnum \count11 > 0
	%
	%%%\com{%%%%%}
	%
	For any $a (\in [1, m])$, any job $J$, 
	any job $J'$ which is given before $J$, and 
	any point $t \in [ r(J'), d(J') ]$, 
	there exists the point $t'$ such that 
	$h_{a}(J, J', t) = t'$ 
	and 
	$t' \in [ r(J'), d(J') ]$ 
	immediately after $J$ is placed. 
	\fi
\end{LMA}
\begin{proof}
	\ifnum \count10 > 0
	%
	%%%\com{%%%%%%}
	%
	$J, J', t$%%%%%%%%%%%%%%%%%%
	$GR$%%%%%%
	$GR$%%%%$J'$%place%%%%%
	%%%%%%%%%%%%%$m_{GR}(J')$%%%%%%
	%%%%%
	$J'$%%%%%%%%%%place%%%%%%%intersect%%%%%%%%%%%%%%%%%
	%
	%
	%%%%
	$J'$%%%%intersect%%
	%$a$th%%%%$p$-interval%%%%%
	$J'$%$oe$-interval%$n$-interval%%%%%%%%%
	%
	%
	%%%%%
	$J'$%%%%intersect%%
	%$a$th%%%%$p$-interval%%%%
	$J'$%$n$-interval%%%%%%%%%%%%%
	$J'$%$oe$-interval%%%%%%%%
	%
	%
	%%%%
	$P_{a}( J, J', r(J'), t' ) = E(J', t)$
	$t' \in [ r(J'), d(J') ]$
	%%%%%$t'$%%%%%%
	%
	
	%
	\fi
	\ifnum \count11 > 0
	%
	%%%\com{%%%%%}
	%
	Suppose that $J, J'$ and $t$ satisfy the conditions of the statement of this lemma. 
	By the definition of $GR$, 
	$GR$ chooses the machine $m_{GR}(J')$ when placing $J'$ so that it gains the largest profit from $J'$. 
	That is, 
	$GR$ chooses the machine so that the total length of the intervals of jobs which were already placed before placing $J'$ 
	and which are intersecting with the interval of $J'$ is minimized. 
	%%%\com{%%%%%%%}
	%
	%
	Hence, 
	the total length of $p$-intervals on the $a (\in \{ 1, m \})$th machine which are intersecting with the interval of $J'$ is 
	at least the length of $oe$-intervals or $n$-intervals of $J'$. 
	Namely, 
	the total length of $p$-intervals on the $a$th machine which are not in $n$-intervals of $J'$ and are intersecting with the interval of $J'$ is 
	at least the length of $oe$-intervals of $J'$. 
	Therefore, 
	there exists the point $t'$ such that 
	$P_{a}( J, J', r(J'), t' ) = E(J', t)$
	and 
	$t' \in [ r(J'), d(J') ]$. 
	\fi
\end{proof}
\ifnum \count10 > 0
%
%%%%%\com{%%%%%%}
%

%

%
\fi
\ifnum \count11 > 0
%
%%%%%\com{%%%%%}
%
\begin{LMA} \label{LMA:ass}
	\ifnum \count10 > 0
	%
	%%%\com{%%%%%%}
	%
	Case~2.2%%%%%%%%%
	%
	%
	%%%%%%
	$oe$-fraction $f$%%%%Case~2.2%%%%%%%
	$f$%%%$p$-fraction $f_{a}$%%%%%%%%%%%%%
	%%%%
	$f_{a}$%%
	%%$a (\in \{ 1, m \})$th%%%%%$t_{a}$%%%%%%%
	Case~2.2%%%%%$M^{-1}_{3}(f_{a}) = \varnothing$%%%%%%
	\fi
	\ifnum \count11 > 0
	%
	%%%\com{%%%%%}
	%
	Case~2.2 is executable. 
	That is, 
	when Case~2.2 is executed for an $oe$-fraction $f$, 
	$f$ can be assigned a $p$-fraction $f_{a}$ 
	such that 
	$M^{-1}_{3}(f_{a}) = \varnothing$ 
	immediately before executing Case~2.2. 
	\fi
\end{LMA}
\begin{proof}
	\ifnum \count10 > 0
	%
	%%%\com{%%%%%%}
	%
	%%%%%%%%%%%$J'$%%
	$J'$%%%%%%%%%%%%%%%$J$%
	$J$%$oe$-fraction $f$%%%%%
	Step~2%%%%%%%%%
	%
	%%%%
	$f$%%$t \in [ r(J), d(J) ]$%%%%%%%%%
	%%%%
	%
	%%%
	$t_{1} = h_{1}(J', J, t)$%%%. 
	%
	
	%
	%%%Step~2%%%%%%%%
	%%%%%%%%%%$t_{1}$%%%%%$p$-fraction%%%%%%%%
	%
	%
	$J'$%place%%%$t_{1}$%%%%%$c$-fraction ($oe$-fraction, $ge$-fraction)%%%%%%%%%%%
	$x$ ($y$, $z$)%%%%
	$OPT$%$t_{1}$%%%%%$c$-fraction%$oe$-fraction%%%%%%%%%place%%%%%%%
	$x + y$%%%%%%%%%%%%%%
	$GR$%$t_{1}$%%%%%%%%%%$x + z$%%%%%%%%%
	$w$%$GR$%$t_{1}$%%%%%place%%%%%%%%%%%%%%%
	\begin{equation} \label{LMA:peqx.eq.1}
		w \geq \max \{ (x + y) - (x + z), 0 \} = \max \{ y - z, 0 \}
	\end{equation}
	%%%%%%
	%
	%
	$h_{1}$%%%%%%
	1%%%%%%%%%2%%%%%%%%$oe$-fraction%%%%$h_{1}$%%%%%%
	%
	%%%%
	%%%%%%$oe$-fraction%%%%%%%%%%$h_{1}$%%%%%%%%%%
	%%%%%
	%%%%%%%%%%%%%%$t_{1}$%%%Step~2%%%%%%%%%%1%%%%%%%
	%
	%
	%%%
	$h_{1}$%%%%%%
	%%%%%%%%%$t_{1}$%%%%%%%%%
	%%%
	%%%%$t_{1}$%%%%%%%$n$-interval%%%%%%%
	%
	%
	%%%%
	%%%%%%%%%%$t_{1}$%%%%%$p$-fraction%%%%
	$t_{1}$%%%%%%%%%
	$t_{1}$%%%%%%%$n$-interval%%%%%%%%%%%%%%%%%%%
	%%%%%%%$x + y + z$%%%%
	%
	
	%
	%%%
	Step~2%%%%%%%%%%%%%%%%%%%%$p$-fraction%%%%%%%%
	1%%$c$-fraction%1%%$ge$-fraction%%%%%
	$oe$-fraction%%%%%%%%%%%%%%%
	%%%%1%%2%%%%
	(Case~2.1)%
	%
	%%%%
	%%%%%%%%%%
	$t_{1}$%%%%%$c$-fraction%$ge$-fraction%%%
	%%%%%$x + 2z$%%%%
	%
	%
	%%%
	$t_{1}$%$GR$%%%%%%%%%%%%$a$%%%%%%%%%
	%%%%%$t_{a} = h_{a}(J', J, t)$%%%$p$-fraction%$oe$-fraction%%%%%%
	(Case~2.2)%
	%
	%
	%%%%%%%%%%
	$h_{a}$%%%%%%
	%%%%%%$p$-fraction%%%%%%%$h_{a}$%%%%%%%%%%
	%%%%%$J$%$oe$-fraction%%%$p$-fraction%%%%%%%%
	%
	%
	%%%%%%%%%%%%
	Step~2%%%%%%%%%%%%%%%$p$-fraction%%%%%%%%%$x + 2z + w$%%%%
	%
	%
	%%%%%
	\[
		x + 2z + w \geq x + 2z + y - z = x + y + z
	\]
	%%%%%%
	%%%%
	%%%%%(\ref{LMA:peqx.eq.1})%%%%
	%
	%
	%%%%
	$f$%%%%Case~2.2%%%%%%%%
	Case~2.2%%%%%%%%%
	\fi
	\ifnum \count11 > 0
	%
	%%%\com{%%%%%}
	%
	Suppose that 
	the routine executes Step~2 for a placed job $J'$, 
	a job $J$ given before $J'$ and an $oe$-fraction $f$ of $J$. 
	Also, 
	suppose that 
	$f$ exists at a point $t \in [ r(J), d(J) ]$ and 
	$t_{1} = h_{1}(J', J, t)$ holds. 
	First, 
	we evaluate the number of $p$-fractions located at $t_{1}$ which are required for the execution of Step~2. 
	Let $x$ ($y$, $z$) denote the number of jobs whose $c$-fractions ($oe$-fractions, $ge$-fractions) located at $t_{1}$ immediately after $J'$ is placed. 
	Since $OPT$ places jobs whose $c$-fractions or $oe$-fractions are located at $t_{1}$, the number of machines is at least $x + y$. 
	The number of machines on which $GR$'s jobs are placed at $t_{1}$ is $x + z$. 
	Let $w$ denote the number of machines on which $GR$ places no jobs at $t_{1}$. Then, 
	\begin{equation} \label{LMA:peqx.eq.1}
		w \geq \max \{ (x + y) - (x + z), 0 \} = \max \{ y - z, 0 \}. 
	\end{equation}
	By definition, 
	the value of $h_{1}$ for an $oe$-fraction of a job is different from that for another $oe$-fraction of the job 
	(i.e., the two $oe$-fractions are located at distinct points). 
	Hence, 
	there is a one-to-one correspondence between $oe$-fractions of a job and the values of $h_{1}$ for the job. 
	That is, 
	the number of executions of Step~2 for a job with $t_{1}$ is exactly one. 
	Further, 
	the interval of such job must include $t_{1}$ and 
	$t_{1}$ is not in an $n$-interval of the job by the definition of $h_{1}$. 
	Therefore, 
	the number of $p$-fractions which are located at $t_{1}$ and required for the assignments is 
	at most the number of jobs 
	such that the interval of each of the jobs includes $t_{1}$ and 
	$t_{1}$ is not in an $n$-interval of each of them, 
	that is, at most $x + y + z$. 
	Second, 
	we evaluate the number of $p$-fractions which can be used for the assignments at the execution of Step~2. 
	The numbers of one $c$-fraction and one $ge$-fraction 
	which can be assigned to $oe$-fractions are one and two, respectively 
	(in Case~2.1). 
	Hence, 
	the number of $c$-fractions and $ge$-fractions which are located at $t_{1}$ and used for the assignments is at least $x + 2z$. 
	Also, 
	for the $a$th machine on which $GR$ places no jobs at the execution of Step~2, 
	the routine assigns a $p$-fraction located at the point $t_{a} = h_{a}(J', J, t)$ (in Case~2.2). 
	In the same way as the above argument, 
	by the definition of $h_{a}$, 
	there is a one-to-one correspondence between $p$-fractions of a job and the values of $t_{a}$ for the job 
	and thus, 
	there is a one-to-one correspondence between the $oe$-fractions of $J$ and the $p$-fractions of $J$. 
	By summing up the above numbers, 
	the number of $p$-fractions for the assignments at the execution of Step~2 is at least $x + 2z + w$. 
	Thus, we have 
	\[
		x + 2z + w \geq x + 2z + y - z = x + y + z, 
	\]
	in which the inequality follows from Eq.~(\ref{LMA:peqx.eq.1}). 
	Therefore, 
	when the routine executes Case~2.2 for $f$, 
	it is executable. 
	\fi
\end{proof}
%

%%%%%%%%%%%%%%%%%%%%%%%%%%%%%%%%%%%%%%%%%%%%%%%%%%%%%%
\subsection{Upper Bound for $m = 2$}
\ifnum \count10 > 0
%
%%%\com{%%%%%%}
%
$m=2$%%%%%$oe$-fraction%%%%%
$p$-fractions%%%%%%%%%%%%
$GR$%%%%%%%%%%
%
%%%%%%%%%%%%%%%%%%%%%%
%%%$m$%%%%%%%%%%$GR$%%%%%%%%%%%
%
%
1%%$oe$-fraction%1%%$ge$-fraction%%%%%%%%%%{\em $ge$-assign}%%%%%%
%%%
1%%$oe$-fraction%3%%$p$-fraction%%%%%%%%%%{\em $3p$-assign}%%%%%%
%%%\com{%%%%%%%%%%}
%
%
%%%%%%%%%%%%%%%%%%%%%%%%%3%%%%%%%%%%%%%%%
%
\begin{description}
\itemsep=-2.0pt
\setlength{\leftskip}{10pt}
\item[1.]
%%%$oe$-fraction%%%%%
$ge$-assign%%%%%%
%%%%$3p$-assign%%%%%%
%
\item[2.]
$GR$%%%%%%%%%%
%%%$c$-fraction%%%1%$3p$-assign%%%%
%
\item[3.]
%%%$ge$-fraction%%
%%1%$ge$-assign%%%%%
%%1%$3p$-assign%%%%%%
%
\end{description}
%
%
%%%%%%%%%%%%%%%%%%
$oe$-interval%%$p$-interval%%2%%%%$N_{1},N_{2}$%%%%%%%%%%
%
%%%
%%%$oe$-fraction $f$%%%%$i \in \{ 1, 2 \}$%%%%
$N_{i}(f) = \varnothing$%%%%%%
%
%%%%%%
%%%$oe$-fraction $f$%%%%%
$N_{1}(f) \cup N_{2}(f) \ne \varnothing$%
$N_{1}(f) \ne \varnothing$%%%%
%%%3%%$p$-fractions $f_{1}, f_{2}, f_{3}$%%%%%%
$N_{1}(f) = \{ f_{1}, f_{2}, f_{3}  \}$%%%%%%
$N_{2}(f) \ne \varnothing$%%%%
$ge$-fraction $f'$%%%%%%
$N_{2}(f) = f'$%%%%%%
%
%
%%%$oe$-fraction $f$%$f'(\ne f)$%%%%$i \in \{ 1, 2 \}$%%%%%
$N_{i}(f) \cap N_{i}(f') = \varnothing$%%%%%%
$3p$-assign%%%$oe$-interval%%%%$V_{\overline{oe}}(\sigma)$%%%%
$ge$-assign%%%$oe$-interval%%%%$V_{\overline{oe'}}(\sigma)$%%%%
%
%
%%%%%%
\[
	V_{\overline{oe}}(\sigma) \leq V_{GR}(\sigma) / 3
\]
\[
	V_{\overline{oe'}}(\sigma) \leq V_{ge}(\sigma)
\]
%%%%%%
%
%
%%%%%%2%%%%%%%%
\begin{eqnarray*}
	V_{OPT}(\sigma) 
		&=& V_{c}(\sigma) + V_{\overline{oe}}(\sigma) + V_{\overline{oe'}}(\sigma) \\ %%% \mbox{\hspace{6mm} (%~(\ref{sec:upanym.eq.3})%(\ref{sec:upanym.eq.4})%%)} \\
		&\leq& V_{c}(\sigma) + V_{GR}(\sigma) / 3 + V_{ge}(\sigma) 
		= \frac{4}{3} V_{GR}(\sigma) \\
%%%		= 3 (V_{c}(\sigma) + V_{ge}(\sigma)) \mbox{\hspace{6mm} (%~(\ref{sec:upanym.eq.1})%%)} \\
%%%		&=& 4 V_{GR}(\sigma) / 3
\end{eqnarray*}
%%%%%%
%
%
%%%%%%%%%%%%
%

%
\fi
\ifnum \count11 > 0
%
%%%\com{%%%%%}
%
When $m=2$, 
we also evaluate the competitive ratio of $GR$ by assigning $p$-fractions to all $oe$-fractions. 
In this case, 
we obtain a better upper bound on the competitive ratio of $GR$ than one for general $m$ by implementing more detailed assignments. 
If the routine assigns one $ge$-fraction to one $oe$-fraction, 
we say that the routine {\em $ge$-assigns} the $ge$-fraction to the $oe$-fraction. 
Also, 
if the routine assigns three $p$-fractions to one $oe$-fraction, 
we say that the routine {\em $3p$-assigns} each of the $p$-fractions to the $oe$-fraction. 
We will show the following three properties by the assignments according to the routine defined later: 
\begin{description}
\itemsep=-2.0pt
\setlength{\leftskip}{10pt}
\item[1.]
Each $oe$-fraction is $ge$-assigned or $3p$-assigned, 
\item[2.]
a $c$-fraction of a job given to $GR$ is $3p$-assigned at most once, and 
\item[3.]
a $ge$-fraction is $ge$-assigned at most once and 
is $3p$-assigned at most once. 
\end{description}
We will show them by sequentially constructing two functions $N_{1}$ and $N_{2}$ from $oe$-intervals to $p$-intervals satisfying the following properties: 
Initially, 
for any $oe$-fraction $f$ and any $i \in \{ 1, 2 \}$, 
$N_{i}(f) = \varnothing$. 
At the end of the input, 
for any $oe$-fraction $f$, 
$N_{1}(f) \cup N_{2}(f) \ne \varnothing$. 
There exist three distinct $p$-fractions $f_{1}, f_{2}$ and $f_{3}$ 
such that $N_{1}(f) = \{ f_{1}, f_{2}, f_{3}  \}$ 
if $N_{1}(f) \ne \varnothing$. 
There exists a $ge$-fraction $f'$ 
such that $N_{2}(f) = f'$ 
if $N_{2}(f) \ne \varnothing$. 
For any $oe$-fractions $f$ and $f'(\ne f)$ and any $i \in \{ 1, 2 \}$, 
$N_{i}(f) \cap N_{i}(f') = \varnothing$. 
Let $V_{\overline{oe}}(\sigma)$ denote the total length of $oe$-intervals to which the routine $3p$-assigns, 
and let $V_{\overline{oe'}}(\sigma)$ denote the total length of $oe$-intervals to which the routine $ge$-assigns. 
Thus, 
\[
	V_{\overline{oe}}(\sigma) \leq V_{GR}(\sigma) / 3
\]
and 
\[	
	V_{\overline{oe'}}(\sigma) \leq V_{ge}(\sigma). 
\]
Then, using these inequalities, we have 
\begin{eqnarray*}
	V_{OPT}(\sigma) 
		&=& V_{c}(\sigma) + V_{\overline{oe}}(\sigma) + V_{\overline{oe'}}(\sigma) \\ %%% \mbox{\hspace{6mm} (%~(\ref{sec:upanym.eq.3})%(\ref{sec:upanym.eq.4})%%)} \\
		&\leq& V_{c}(\sigma) + V_{GR}(\sigma) / 3 + V_{ge}(\sigma) 
		= \frac{4}{3} V_{GR}(\sigma). 
%%%		= 3 (V_{c}(\sigma) + V_{ge}(\sigma)) \mbox{\hspace{6mm} (%~(\ref{sec:upanym.eq.1})%%)} \\
%%%		&=& 4 V_{GR}(\sigma) / 3
\end{eqnarray*}
Therefore, 
we have the following theorem: 
\fi
\begin{theorem}\label{thm:upm2}
	\ifnum \count10 > 0
	%
	%%%\com{%%%%%%}
	%
	$m = 2$%%%%
	$GR$%%%%%%%$4/3$%%%%
	\fi
	\ifnum \count11 > 0
	%
	%%%\com{%%%%%}
	%
	When $m = 2$, 
	the competitive ratio of $GR$ is at most $4/3$. 
	\fi
\end{theorem}
\ifnum \count10 > 0
%
%%%\com{%%%%%%}
%
%%%$i \in \{ 1, 2 \}$%%%%$p$-fraction $f'$%%%%%
$N^{-1}_{i}(f') = \{ f \mid N_{i}(f) = f' \}$
%%%%%%
%
%
$N_{1}^{-1}(f') = \varnothing$%%%%%%%%$p$-fraction $f'$%%
{\em 1-assignable}%%%%%%%
$N_{2}^{-1}(f') = \varnothing$%%%%%%%%$ge$-fraction $f'$%%
{\em 2-assignable}%%%%%%%
%
%
%%%2%%%%%%%%%%%%%%%%%%%
%

%
\fi
\ifnum \count11 > 0
%
%%%\com{%%%%%}
%
For any $i \in \{ 1, 2 \}$ and any $p$-fraction $f'$, 
define $N^{-1}_{i}(f') = \{ f \mid N_{i}(f) = f' \}$. 
We say that a $p$-fraction $f'$ is {\em 1-assignable} 
if $N_{1}^{-1}(f') = \varnothing$. 
Also, 
we say that a $ge$-fraction $f'$ is {\em 2-assignable} 
if $N_{2}^{-1}(f') = \varnothing$. 
Now we give the definition of the routine to construct the above two functions. 
\fi
\ifnum \count10 > 0
%
%%%\com{%%%%%%\\}
%
\noindent\vspace{-1mm}\rule{\textwidth}{0.5mm} 
\vspace{-3mm}
{{\sc AssignmentRoutine2}}\\
\rule{\textwidth}{0.1mm}
%
	%
	%%%$J$%place%%%%%%%%%%
	%
	$J$%%$oe$-fraction $f$%%%%%%%%%%%%
	\\
	{\bf\boldmath Step~1:}
	$m_2 := m_{GR}(J)$%%%
	$m_1 := \{ 1, 2 \} \setminus \{ m_2 \}$%%%%
	$t_{1} := h_{m_1}(J, {J}, t)$%
	%%%%$f$%%$t$%%%%%%%%%
	\\
	{\bf\boldmath Step~2:}
		$m_{2}$%%%%%%%%$t$%%%%%$p$-fraction%$f'$%%%
		%$oe$-fraction%%%%%%$f'$%%%%%%%
		%%2%%Case%%%%%%%%%\\
%
	\hspace*{1mm}
	{\bf\boldmath Case~2.1 ($f'$%2-assignable%%%%%):} 
		$N_{2}(f) := f'$%
		\\
	\hspace*{1mm}
	{\bf\boldmath Case~2.2 (%%%%%%%):}
		$N_{1}(f) := \{ f', f_{1}, f_{2} \}$%
		%%%%
		$f_{1}$%$GR$%$m_1$th%%%%%$t_{1}$%%%%%$p$-fraction%%%
		%%%~\ref{LMA:peqx}%%%$f_{1}$%%%%%%%
		%%%
		$f_{2}$%$GR$%$m_2$%%%%%$t_{1}$%%%%%$p$-fraction%%%
		%$h_{m_1}$%%%%%%$J$%%%%$t_{1}$%%%%%%
		$f_{2}$%%%%%%%
		%%%~\ref{LMA:upm2}%%%
		$f', f_{1}, f_{2}$%1-assignable%%%%%
		\\
\noindent\vspace{-1mm}\rule{\textwidth}{0.5mm} 
\fi
\ifnum \count11 > 0
%
%%%\com{%%%%%\\}
%
\noindent\vspace{-1mm}\rule{\textwidth}{0.5mm} 
\vspace{-3mm}
{{\sc AssignmentRoutine2}}\\
\rule{\textwidth}{0.1mm}
	Consider a moment immediately after a job $J$ is placed. 
	For any $oe$-fraction $f$ of ${J}$, 
	execute the following. 
	\\
	{\bf\boldmath Step~1:}
	$m_2 := m_{GR}(J)$ and 
	$m_1 := \{ 1, 2 \} \setminus \{ m_2 \}$. 
	$t_{1} := h_{m_1}(J, {J}, t)$, 
	in which $f$ exists at a point $t$. 
	\\
	{\bf\boldmath Step~2:}
		Let $f'$ be the $p$-fraction at $t$ on the $m_{2}$th machine 
		($f'$ exists by the definition of $oe$-fractions). 
		Execute one of the following two cases. \\
	\hspace*{1mm}
	{\bf\boldmath Case~2.1 ($f'$ is 2-assignable):} 
		$N_{2}(f) := f'$. 
		\\
	\hspace*{1mm}
	{\bf\boldmath Case~2.2 (Otherwise):}
		$N_{1}(f) := \{ f', f_{1}, f_{2} \}$, 
		in which 
		$f_{1}$ is the $p$-fraction at $t_{1}$ on $GR$'s $m_1$th machine
		($f_{1}$ exists by Lemma~\ref{LMA:peqx}), and 
		$f_{2}$ is the $p$-fraction at $t_{1}$ on $GR$'s $m_2$th machine
		($f_{2}$ exists because the interval of $J$ contains $t_{1}$ by the definition of $h_{m_1}$). 
		(By Lemma~\ref{LMA:upm2}, 
		$f', f_{1}$ and $f_{2}$ are 1-assignable.)
		\\
\noindent\vspace{-1mm}\rule{\textwidth}{0.5mm} 
\fi

\begin{LMA} \label{LMA:upm2}
	\ifnum \count10 > 0
	%
	%%%\com{%%%%%%}
	%
	Case~2.2%%%%%%%%%
	%
	%%%%%%
	$oe$-fraction $f$%%%%Case~2.2%%%%%%%
	$f$%%%3%$p$-fraction%%%%%%%%%
	%%%%%$3p$-assign%%%%%
	%%%%
	%%%%%$p$-fraction%Case~2.2%%%%%1-assignable%%%%
	%
	
	%
	\fi
	\ifnum \count11 > 0
	%
	%%%\com{%%%%%}
	%
	Case~2.2 is executable. 
	That is, 
	when Case~2.2 is executed for an $oe$-fraction $f$, 
	$f$ can be assigned 3 $p$-fractions 
	(i.e., $3p$-assigned)
	each of which is 1-assignable immediately before executing Case~2.2. 
	\fi
\end{LMA}
\begin{proof}
	\ifnum \count10 > 0
	%
	%%%\com{%%%%%%}
	%
	%%%%%%%%%%%%%%%%%%%%%%%%%%
	%
	%%%%%%%%%%%%%%%%%%%%%%%%%%
	%%%%
	%%%$J$%%%%%%%%%%Case~2.2%%%%%%%%%%
	$J$%%%%%%%%%%%%%%%%%%
	%
	%
	%%%%
	%%$t$%$J$%$oe$-fraction $f$%%%%Step~2%%%%%%%%%%%%
	%
	$m_{GR}(J)$%%%%%%%%%$t$%%%%%%
	%%%$J'$%$p$-fraction%$f'$%%%%
	%
	%
	%%%
	$t$%%%%%$p$-fraction%$3p$-assign%%%%%%%$J''$%$oe$-fraction $f''$%%%%%%%%
	$J''$%$J$%%%%%%%%%%
	$f''$%%$t''$%%%%%%%%%
	{\bf (i) $t'' \ne t$%%%%%%%:}
	$t'' \ne t$%%%%
	$t = h_{m'_{1}}(J'', J'' t'')$%%%%%%
	%%%%$m'_{1} = \{ 1, 2 \} \setminus m_{GR}(J'')$%%%%%%
	%
	%
	%%%%$\hat{J} (\ne J'')$%%
	$J$%%%%%%%%%$J''$%%%%%%%%%%%%%%%%
	Case~2.2%$J''$%%%%%%%%%%%
	$3p$-assign%%%%%%%%%%%$t$%%%%%%%%%%%%%%
	%
	%
	%%%%
	$\hat{J}$%$t$%%%%%%%%%
	%%%%%$n$-interval%%%%
	%
	%
	%%%%
	Case~2.2%%%%%%
	$\hat{J}$%$t$%%%%%%$oe$-fraction%%$t$%$p$-fraction%$3p$-assign%%%%%%%%%
	%
	%
	%%%%%%%%%
	%%%%%Case~2.2%%%
	$3p$-assign%%%%%%$t$%%%%%$p$-fraction%%%2%%%%%
	%
	%
	%%%%%$t$%%%%%%%%%%%%%3%($J,J',J''$)%%%
	$m = 2$%%%%%%
	$t$%%$ge$-fraction%1%%%%%%%%
	{\bf (ii) $t'' = t$%%%%%%%:}
	$t$%$f$%$f''$%%%%%%%%
	$m=2$%%%
	$t$%$ge$-fraction%1%%%%%%%%
	%
	%%%$t$%$oe$-fraction%$f$%$f''$%2%%%%%%%
	%
	%
	%%%%%
	$f''$%%%%%%
	Case~2.1%%%%%%$ge$-assign%%%%%%%%
	Case~2.2%%%%%%%%
	%
	%%%%$t$%%%%%$p$-fraction%%
	$J$%%%%%%%%%%%%%$t$%%%%%$oe$-fraction%$3p$-assign%%%%%%%%
	%
	%
	%%%%
	$f$%$1$-assignable%%%%
	%
	
	%
	%%%$f$%%%%%%%%%Case~2.2%%%%%%%%%%%%%%
	%
	%
	%%$t$%%%%%$p$-fraction%$ge$-assign%%%%%%%%%%%%
	%
	%
	$t$%%%%%%%%%%3%%%%%%
	$m=2$%%%%%%
	$t$%%$oe$-fraction%%%%%$ge$-fraction%%%%%%
	Case~2.1%%%%%%
	$t$%%%$oe$-fraction%$t$%$ge$-fraction%$ge$-assign%%%%%%%%
	%
	%%%%
	$2$-assignable%$ge$-fraction%$t$%%%%%%%%1%%%%%%%%
	%%%%%$f$%%%%Case~2.1%%%%%%%
	%
	
	%
	%%%
	$t$%%%%%$p$-fraction%$ge$-assign%%%%%%%%%%%%%
	$t$%%%%%$p$-fraction%$3p$-assign%%%%%%%%%%%%
	%
	%%%(i)%%%
	$t$%%$ge$-fraction%%%%%%%%
	$J$%%%%%%%%%
	%%$ge$-fraction%%%$ge$-assign%%%%%%%%%%%
	$f$%%%%Case~2.1%$ge$-assign%%%%
	%
	
	%
	%%%%
	$t$%%%%%$p$-fraction%$3p$-assign%%%%%%%%%%%%%
	%
	%%%%%
	$J$%%%%%%%%%$f'$%$1$-assignable%%%%
	$t_{1} = h_{m_{1}}(J, J, t)$%%%%
	%%%%
	$m_{1} = \{ 1, 2 \} \setminus m_{GR}(J)$
	%%%%%%
	%
	%
	(i)%%%
	$t_{1}$%$p$-fraction%$3p$-assign%%%%%%%%$oe$-fraction%%
	$J$%$f$%%%%%%
	%%%%%%
	%%%%%%%%%$oe$-fraction%%
	$t_{1}$%%%%%%%%%%%$n$-interval%%%%%
	(ii)%%%
	$t_{1}$%%%%%$oe$-fraction%$t_{1}$%$p$-fraction%$3p$-assign%%%%%
	%
	%%%%
	$t_{1}$%$p$-fraction%$1$-assignable%%%%
	$f'$%$1$-assignability%%%%%$f$%%%%Case~2.2%%%%%%%%%
	\fi
	\ifnum \count11 > 0
	%
	%%%\com{%%%%%}
	%
	We prove the lemma by induction on the number of given jobs. 
	The statement of the lemma is clearly true before the first job is given. 
	We assume that Case~2.2 is executable immediately before a job $J$ is given, and show that it is executable for $J$ as well. 
	Then, 
	suppose that the routine executes Step~2 for an $oe$-fraction $f$ of a job $J$ at a point $t$. 
	Let $f'$ be the $p$-fraction of a job $J'$ which is given before $J$ at $t$ on the $m_{GR}(J)$th machine. 
	First of all, 
	consider an $oe$-fraction $f''$ of a job $J''$ to which a $p$-fraction at $t$ is $3p$-assigned. 
	Suppose that $f''$ is located at a point $t''$. 
	{\bf (i) ($t'' \ne t$):}
	Since $t'' \ne t$, 
	$t = h_{m'_{1}}(J'', J'' t'')$, 
	in which $m'_{1} = \{ 1, 2 \} \setminus m_{GR}(J'')$. 
	Suppose that $\hat{J} (\ne J'')$ is a job given before $J$ and after $J''$. 
	When Case~2.2 is executed for $J''$, 
	there exists jobs including $t$ on the both machines, 
	which are used for $3p$-assignments to $J''$. 
	It follows that 
	if the interval of $\hat{J}$ contains $t$, 
	then $t$ is in an $n$-interval of $\hat{J}$. 
	Thus, 
	an $oe$-fraction not at $t$ of $\hat{J}$ is not $3p$-assigned $p$-fractions at $t$ by the definition of Case~2.2 of the routine. 
	Hence, in this case, 
	the number of $p$-fractions at $t$ used for $3p$-assignments is exactly two 
	by the definition of Case~2.2. 
	Then, 
	there exist at least three jobs (i.e., $J,J'$ and $J''$) whose intervals contain $t$ 
	and $m=2$, 
	which means that 
	there exists at least one $ge$-fraction at $t$. 
	{\bf (ii) ($t'' = t$):}
	Since both $f$ and $f''$ exist at $t$ and 
	$m=2$, 
	at least one $ge$-fraction exists at $t$. 
	Also, 
	all the $oe$-fractions at $t$ are only $f$ and $f''$. 
	Thus, 
	the routine $ge$-assigns to $f''$ by executing Case~2.1. 
	That is, Case~2.2 is not executed. 
	Hence, 
	the routine does not $3p$-assign $p$-fractions at $t$ to an $oe$-fraction of a job given before $J$ at $t$. 
	Thus, 
	$f$ is $1$-assignable. 
	Now we are ready to show that 
	it is possible for the routine to execute Case~2.2 for $f$. 
	First, 
	we discuss the case in which $p$-fractions at $t$ are $ge$-assigned. 
	Since the number of jobs whose intervals contain $t$ is at least three and $m=2$, 
	the number of $ge$-fractions at $t$ is equal to that of $oe$-fractions at $t$. 
	By the definition of Case~2.1, 
	$oe$-fractions not at $t$ are not $ge$-assigned $ge$-fractions at $t$. 
	Hence, 
	there exists at least one $ge$-fraction at $t$ which is $2$-assignable and 
	Case~2.1 is executed for $f$. 
	Second,  
	we discuss the case in which $p$-fractions at $t$ are not $ge$-assigned. 
	We first consider the case in which $p$-fractions at $t$ are $3p$-assigned. 
	By the above discussion (i), 
	there exists at least one $ge$-fraction at $t$. 
	When $J$ is given, 
	the $ge$-fraction has not been used for $ge$-assignments yet and 
	the routine $ge$-assigns to $f$ in Case~2.1.  
	Finally, 
	we consider the case in which $p$-fractions at $t$ are not $3p$-assigned.
	In this case, 
	$f'$ is $1$-assignable when $J$ is given. 
	Let $t_{1} = h_{m_{1}}(J, J, t)$, 
	in which 
	$m_{1} = \{ 1, 2 \} \setminus m_{GR}(J)$. 
	By the above (i), 
	an $oe$-fraction which can be $3p$-assigned $p$-fractions at $t_{1}$ is only $f$ of $J$ 
	(i.e., 
	even if the interval of a job with another $oe$-fraction contains $t_{1}$, 
	$t_{1}$ is in an $n$-interval of the job).  
	An $oe$-fraction at $t_{1}$ is not $3p$-assigned a $p$-fraction at $t_{1}$ by above (ii). 
	Therefore, 
	$p$-fractions at $t_{1}$ are $1$-assignable and 
	Case~2.2 is executable for $f$ together with the $1$-assignability of $f'$. 
	\fi
\end{proof}
\ifnum \count10 > 0
%
%%%\com{%%%%%%}
%
%%%%%%%%%%%%%%%%%%%%%%%
%

%
\fi
\ifnum \count11 > 0
%
%%%\com{%%%%%}
%
We show that our analysis of $GR$ for $m = 2$ is tight in the following theorem. 
\fi
\begin{theorem}\label{thm:VL.LB2GR}
	\ifnum \count10 > 0
	%
	%%%\com{%%%%%%}
	%
	$m = 2$%%%%
	%%%$\varepsilon > 0$%%%%%
	$GR$%%%%%%%%%%$4/3 - \varepsilon$%%%%
	\fi
	\ifnum \count11 > 0
	%
	%%%\com{%%%%%}
	%
	When $m = 2$, 
	for any $\varepsilon > 0$, 
	the competitive ratio of $GR$ is at least $4/3 - \varepsilon$. 
	\fi
\end{theorem}
\begin{proof}
	\ifnum \count10 > 0
	%
	%%%\com{%%%%%%}
	%
	%%%%$\sigma$%%%%%
	%
	1%%%%%%$J_1$%%
	$r(J_1) = 0$%$d(J_1) = 1$%%%%%%%
	2%%%%%%$J_2$%%
	$r(J_2) = 2 - \epsilon$%$d(J_2) = 2$%%%%%%%
	%%%%$0 < \epsilon < 1$%%%%%%
	3%%%%%%$J_3$%%
	$r(J_3) = 1$%$d(J_3) = 2$%%%%%%%
	4%%%%%%$J_4$%%
	$r(J_4) = 0$%$d(J_4) = 2$%%%%%%%
	$GR$%%
	$J_1$%$J_2$%1%%%%%%%%
	$J_3$%$J_4$%2%%%%%%%%%%%%%
	$OPT$%%
	$J_1$%$J_2$%$J_3$%1%%%%%%%%
	$J_4$%2%%%%%%%%%%%%%
	%
	%
	%%%%%
	\[
		\frac{ V_{OPT}(\sigma) }{ V_{GR}(\sigma) } 
			= \frac{ 4 }{ 3 + \epsilon } = \frac{4}{3} - \varepsilon
	\]
	%%%%%%
	%%%%
	$\varepsilon = 4 \epsilon / (9 + 3\epsilon)$%%%%%
	\fi
	\ifnum \count11 > 0
	%
	%%%\com{%%%%%}
	%
	Consider the following input $\sigma$. 
	The first job $J_1$ such that 
	$r(J_1) = 0$ and $d(J_1) = 1$ is given. 
	The second job $J_2$ such that 
	$r(J_2) = 2 - \epsilon$ and $d(J_2) = 2$ is given, 
	where $0 < \epsilon < 1$. 
	The third job $J_3$ such that 
	$r(J_3) = 1$ and $d(J_3) = 2$ 
	and 
	the fourth job $J_4$ such that 
	$r(J_4) = 0$ and $d(J_4) = 2$
	are given. 
	$GR$ places $J_1$ and $J_2$ on the first machine, and 
	then places $J_3$ and $J_4$ on the second machine. 
	%%%\com{%%ed:puts%places%}
	%
	%
	$OPT$ places 
	$J_1$, $J_2$ and $J_3$ on one machine 
	and 
	places $J_4$ on the other machine. 
	Thus, 
	\[
		\frac{ V_{OPT}(\sigma) }{ V_{GR}(\sigma) } 
			= \frac{ 4 }{ 3 + \epsilon } = \frac{4}{3} - \varepsilon, 
	\]
	in which 
	$\varepsilon = 4 \epsilon / (9 + 3 \epsilon)$. 
	\fi
\end{proof}

%%%%%%%%%%%%%%%%%%%%%%%%%%%%%%%%%%%%%%%%%%%%%%%%%
\section{Lower Bounds for Uniform Profit Case} \label{sec:lbuc}
\ifnum \count10 > 0
%
%%%\com{%%%%%%}
%
%%%%%uniform profit case%%%%
%%%%%%%%%%
%
%
%%%%%%%%%%%%%
$m=2$%%%%%%%%

\fi
\ifnum \count11 > 0
%
%%%\com{%%%%%}
%
In this section, 
we show lower bounds on the competitive ratios of online algorithms for the uniform profit case. 
For better understanding, 
we first consider the case of $m=2$. 
\fi

\begin{theorem}\label{thm:VL.LB}
	\ifnum \count10 > 0
	%
	%%%\com{%%%%%%}
	%
	$m = 2$%%%%
	%%%%%%%%%%%%%%%%%%%%%%%%%%%$(10 - \sqrt{2}) / 7 \geq 1.226$%%%%% 1.22654091966 
	%
	
	%
	\fi
	\ifnum \count11 > 0
	%
	%%%\com{%%%%%}
	%
	When $m = 2$, 
	the competitive ratio of any deterministic online algorithm is at least $(10 - \sqrt{2}) / 7 \geq 1.226$. % 1.22654091966
	\fi
\end{theorem}
\begin{proof}
	\ifnum \count10 > 0
	%
	%%%\com{%%%%%%}
	%
	%%%%%%%%%%%$ON$%%%%%%%%
	%
	%%%job $J_1$%%
	$r(J_1) = 0$%$d(J_1) = 1$%%%%%
	2%%%job $J_2$%%
	$r(J_2) = 1+x$%$d(J_2) = 2+x$%%%%%%
	%%%%$x$%%%%%%%
	%
	%
	%%%%%%%%%%%%%%
	$ON$%$OPT$%$J_1$%%%%%%%%%%%%%%%%%%
	%
	
	%
	%%%%2%%%%%%%%%%
	%
	%
	%%%
	$ON$%$J_1$%$J_2$%%%%%%%%%%%%%%%%%%%
	%%%%%$J_{2}$%%%%%%%%%%%%%%%%%%
	%
	3%%%job $J_3$%$r(J_3) = 0$%$d(J_3) = 2 + x$%%%%%%%
	%
	%%%%%job%%%%%%%%
	%
	%%%%%$\sigma_1$%%%%%%%%%
	%
	%
	$ON$%$J_3$%1st%%%%%%%%%%%%
	$V_{ON}(\sigma_1) = 2 + x + 1 = 3 + x$
	%%%%%%
	%
	%
	$ON$%$J_3$%2nd%%%%%%%%%%%%$ON$%%%%%%%%%
	%
	%
	%%%
	$OPT$%%$J_1$%$J_2$%%%%%%%%%%%%
	$J_3$%%%%%%%%%%%%%
	%
	%%%%
	$V_{OPT}(\sigma_1) = 2 + 2 + x = 4 + x$
	%%%%%%
	%
	%
	%%%%
	\begin{equation} \label{thm:VL.LB.eq.1}
		\frac{ V_{OPT}(\sigma_1) }{ V_{ON}(\sigma_1) } = \frac{ 4 + x }{ 3 + x }
	\end{equation}
	%%%%%%
	%
	
	%
	%%%
	$ON$%$J_1$%$J_2$%1st%%%%%%%%%%%%%%%%
	3%%%job $J'_1$%$r(J'_1) = 1 - y$%$d(J'_1) = 1 + x$%%%%%%%
	4%%%job $J'_2$%$r(J'_2) = 1$%$d(J'_2) = 1 + x + y$%%%%%%%
	%%%%$y$%%%%%%%
	%
	%%%%%job%%%%%%%%
	%
	%%%%%$\sigma_2$%%%%%%%%%
	%
	%
	%%%
	$ON$%$J'_1$%$J'_2$%%%%%%%%%%%%%%%%%%%
	$J'_1$%1st%%%%%%%%%%%%
	\begin{equation} \label{thm:VL.LB.eq.2}
		V_{ON}(\sigma_2) = 1 + x + 1 + x + y = 2 + 2x + y
	\end{equation}
	%%%%%%
	%
	%
	$J'_2$%1st%%%%%%%%%%%%$ON$%%%%%%%%%
	$ON$%$J'_1$%$J'_2$%%%%%%%%%%%%%%%%%%%
	%%%%2nd%%%%%%%%%%%%%%%%
	%%%%%%%%%%%%%%%%%%%%%%%%%%%
	%
	%%%%
	\begin{equation} \label{thm:VL.LB.eq.3}
		V_{ON}(\sigma_2) = 2 + x + 2y
	\end{equation}
	%%%%%%
	%
	%
	%%%%
	$y = x$%%%%%
	%%2%%%~(\ref{thm:VL.LB.eq.2})%(\ref{thm:VL.LB.eq.3})%%%
	$V_{ON}(\sigma_2) = 2 + 3x$
	%%%%%%
	%
	%
	%%%
	$OPT$%%$J_1$%$J'_2$%%%%%%%%%%%%
	$J_2$%$J'_1$%%%%%%%%%%%%%
	%
	%%%%
	$V_{OPT}(\sigma_2) = 2 (1 + x + y) = 2 + 4x$
	%%%%%%
	%
	%
	%
	%%%%
	\begin{equation} \label{thm:VL.LB.eq.4}
		\frac{ V_{OPT}(\sigma_2) }{ V_{ON}(\sigma_2) } = \frac{ 2 + 4x }{ 2 + 3x }
	\end{equation}
	%%%%%%
	%
	
	%
	%~(\ref{thm:VL.LB.eq.1})%(\ref{thm:VL.LB.eq.4})%%%
	\[
		\frac{ V_{OPT}(\sigma) }{ V_{ON}(\sigma) } 
			\geq \min \left \{ \frac{ 4 + x }{ 3 + x }, \frac{ 2 + 4x }{ 2 + 3x } \right \} 
			= \frac{ 4 + \sqrt{2} }{ 3 + \sqrt{2} }
			= \frac{ 10 - \sqrt{2} }{ 7 } 
	\]
	%%%%%%
	%%%%$x = \sqrt{2}$%%%%%%
	%
	
	%
	\fi
	\ifnum \count11 > 0
	%
	%%%\com{%%%%%}
	%
	Consider an online algorithm $ON$. 
	The first given job is $J_1$ such that 
	$r(J_1) = 0$ and $d(J_1) = 1$. 
	The second job is $J_2$ such that 
	$r(J_2) = 1+x$ and $d(J_2) = 2+x$. 
	Note that $x$ is set later. 
	Without loss of generality, 
	we may assume that both $ON$ and $OPT$ place $J_1$ onto the first machine. 
	In the following, 
	we use two inputs. 
	First, 
	we consider the case where 
	$ON$ places $J_1$ and $J_2$ on two different machines.
	%%%\com{%%ed:puts,into%places,on%}
	%%%\com{%%edanz%%%%%%%%%%%%%}
	That is, suppose that $ON$ places $J_{2}$ on the second machine. 
	Then, 
	the third job $J_3$ such that $r(J_3) = 0$ and $d(J_3) = 2 + x$ is given, 
	and no further job arrives.
	We call this input $\sigma_1$. 
	If $ON$ places $J_3$ onto the first machine, 
	we have $V_{ON}(\sigma_1) = 2 + x + 1 = 3 + x$. 
	$ON$ also gains the same profit  
	if $ON$ places $J_3$ onto the second machine. 
	On the other hand, 
	the machine onto which $OPT$ places both $J_1$ and $J_2$ 
	is different from that onto which $J_3$ is placed. 
	Thus, 
	$V_{OPT}(\sigma_1) = 2 + 2 + x = 4 + x$. 
	By the above argument, 
	\begin{equation} \label{thm:VL.LB.eq.1}
		\frac{ V_{OPT}(\sigma_1) }{ V_{ON}(\sigma_1) } = \frac{ 4 + x }{ 3 + x }. 
	\end{equation}
	Second, 
	%%%\com{%%ed:secondly%second%}
	we consider the case where 
	$ON$ places $J_1$ and $J_2$ onto the first machine. 
	The third job $J'_1$ such that $r(J'_1) = 1 - y$ and $d(J'_1) = 1 + x$ and 
	the fourth job $J'_2$ such that $r(J'_2) = 1$ and $d(J'_2) = 1 + x + y$ are given, 
	where $y$ is fixed later. 
	No further job is given; 
	we call this input $\sigma_2$. 
	We first consider the case where 
	$ON$ places $J'_1$ and $J'_2$ on different machines. 
	If $J'_1$ is placed onto the first machine, 
	on which $J_1$ and $J_2$ are placed, 
	\begin{equation} \label{thm:VL.LB.eq.2}
		V_{ON}(\sigma_2) = 1 + x + 1 + x + y = 2 + 2x + y. 
	\end{equation}
	$ON$ gains the same profit 
	if $J'_2$ is placed onto the first machine. 
	Next, 
	we consider the case in which 
	$ON$ places $J'_1$ and $J'_2$ onto the machine. 
	If the machine is the second one, 
	then it is clear that $ON$ gains larger profits 
	than it does in the other case. 
	%%%\com{%%ed:that%it does%}
	%
	Hence, 
	\begin{equation} \label{thm:VL.LB.eq.3}
		V_{ON}(\sigma_2) = 2 + x + 2y.
	\end{equation}
	Now, 
	set $y = x$ and 
	we have 
	$V_{ON}(\sigma_2) = 2 + 3x$
	by Eqs.~(\ref{thm:VL.LB.eq.2}) and (\ref{thm:VL.LB.eq.3}). 
	On the other hand, 
	$OPT$ places both $J_1$ and $J'_2$ onto the first machine 
	and
	both $J_2$ and $J'_1$ onto the second machine. 
	Thus, 
	$V_{OPT}(\sigma_2) = 2 (1 + x + y) = 2 + 4x$. 
	By the above argument, 
	\begin{equation} \label{thm:VL.LB.eq.4}
		\frac{ V_{OPT}(\sigma_2) }{ V_{ON}(\sigma_2) } = \frac{ 2 + 4x }{ 2 + 3x }. 
	\end{equation}
	Therefore, 
	by Eqs.~(\ref{thm:VL.LB.eq.1}) and (\ref{thm:VL.LB.eq.4}), 
	\[
		\frac{ V_{OPT}(\sigma) }{ V_{ON}(\sigma) } 
			\geq \min \left \{ \frac{ 4 + x }{ 3 + x }, \frac{ 2 + 4x }{ 2 + 3x } \right \} 
			= \frac{ 4 + \sqrt{2} }{ 3 + \sqrt{2} }
			= \frac{ 10 - \sqrt{2} }{ 7 }, 
	\]
	where we choose $x = \sqrt{2}$. 
	\fi
\end{proof}
\ifnum \count10 > 0
%
%%%\com{%%%%%%}
%
%%~\ref{thm:VL.LB}%%%%%%%%%%%%%%%%%%%%%%%%
%%%%%%%%%%%$m \geq 3$%%%%%%%%%%%%%%
%

%
\fi
\ifnum \count11 > 0
%
%%%\com{%%%%%}
%
The following theorem provides lower bounds for $m \geq 3$ by generalizing the input used to prove Theorem~\ref{thm:VL.LB}. 
\fi
\begin{theorem}\label{thm:VL.LBm}
	\ifnum \count10 > 0
	%
	%%%\com{%%%%%%}
	%
	%%%$m \geq 3$%%%%%
	%%%%%%%%%%%%%%%%%%%%%%%%%%
	%~\ref{tab:thm:VL.LBm}%%%%%%%%%%
	%
	
	%
	\fi
	\ifnum \count11 > 0
	%
	%%%\com{%%%%%}
	%
%	For each $m \geq 3$, 
%	lower bounds on the competitive ratios of any deterministic online algorithms are shown in Table~\ref{tab:thm:VL.LBm}. 
	The competitive ratio of any deterministic algorithm is at least $1.101$. 
	It is better for fixed $m$ and then refer to Table~\ref{tab:thm:VL.LBm} for details. 
	\fi
\end{theorem}
\begin{proof}
	\ifnum \count10 > 0
	%
	%%%\com{%%%%%%}
	%
	%%%%%%%%%%%$ON$%%%%%%%%
	%
	%%%$\lceil \frac{m}{2} \rceil (= m')$%%job $J_{1,j} \hspace{1mm} (j = 1, \ldots, m')$
	%%%%%%%
	%%%%$r(J_{1,j}) = 0$%$d(J_{1,j}) = 1$%%%%%%
	%%%%job%%%%$S_1$%%%%
	%
	%%
	$m'$%%job $J_{2,j} \hspace{1mm} (j = 1, \ldots, m')$
	%%%%%%%
	%%%%$r(J_{2,j}) = 1 + x$%$d(J_{2,j}) = 2 + x$%%%%%%
	$x$%%%%%%%
	%%%%job%%%%$S_2$%%%%
	%
	%
	$ON$%$S_1$%$S_2$%%%%%%%%%%%1%job%%%%%%%%%%%%%$A$%%%%
	$ON$%
	$S_1$%%%%%%%1%job%%%%%%
	$S_2$%%job%%%%%%%%%%%%%%$B$%%%%
	$ON$%
	$S_2$%%%%%%%1%job%%%%%%
	$S_1$%%job%%%%%%%%%%%%%%$B'$%%%%
	$ON$%
	$S_1$%$S_2$%%%%%%job%%%%%%%%%%%%%%$C$%%%%
	$A,B,B',C$%%%%%%%%%%%%$a, b, b', c$%%%%
	%
	%
	%%%%%
	\begin{equation} \label{thm:VL.LBm.eq.1}
		a + b + b' + c = m
	\end{equation}
	\begin{equation} \label{thm:VL.LBm.eq.1.1}
		a + b \leq m'
	\end{equation}
	\begin{equation} \label{thm:VL.LBm.eq.1.2}
		a + b' \leq m'
	\end{equation}
	%%%%%%
	%
	%
	
	%
	%%%%2%%%%%%%%%%
	%
	%%%%$\sigma_1$%%%%%
	%
	%
	$\lfloor \frac{m}{2} \rfloor (= m'')$%%
	job $J_{3,j} \hspace{1mm} (j = 1, \ldots, m'')$
	%%%%%%%
	%%%%$r(J_{3,j}) = 0$%$d(J_{3,j}) = 2 + x$%%%%%%
	%
	%%%%job%%%%$S_3$%%%%
	%
	
	%
	%%%%
	$a \leq m'$%%%%%%%%
	%(\ref{thm:VL.LBm.eq.1})%%%
	$b + b' + c = m - a \geq m - m' = m''$
	%%%%%%
	%
	%%%%
	$S_3$%%%%job%$B \cup B' \cup C$%%%%%1%%%%%%%%%%%%%%%
	%%%%%%%%%$A$%%%%%%%%%%%%%%%%%%%%%%%%%%%%
	%
	%
	%%%%
	$C$%%%%%$S_3$%%$\min \{ c, m'' \}$%%job%%%%%%%%%
	%
	%%%$\max \{ m'' - c, 0 \}$%%job%%
	$B$%%%%%$B'$%%%%%%%%%%%%%%%
	%
	%
	%%%%
	\begin{eqnarray*} %\label{thm:VL.LBm.eq.2}
		V_{ON}(\sigma_1) 
			&\leq& 2 a + b + b' + (2 + x) \min \{ c, m'' \} + (1 + x) \max \{ m'' - c, 0 \}\\
			&=& a + m - c + (2 + x) \min \{ c, m'' \} + (1 + x) \max \{ m'' - c, 0 \} \mbox{\hspace{6mm} (%~(\ref{thm:VL.LBm.eq.1})%%)} \nonumber
	\end{eqnarray*}
	%%%%%%
	%
	%
	$c \geq m''$%%%%
	\begin{equation} \label{thm:VL.LBm.eq.2.1}
		V_{ON}(\sigma_1) 
			\leq a + m - c + (2 + x) m'' + (1 + x) 0 
			\leq a + m - m'' + (2 + x) m'' 
			= a + m + (1 + x) m''
	\end{equation}
	%%%%%%
	%%%%2%%%%%%%$c \geq m''$%%%%%%%
	%
	%
	$c < m''$%%%%
	\begin{equation} \label{thm:VL.LBm.eq.2.2}
		V_{ON}(\sigma_1) 
			\leq a + m - c + (2 + x) c + (1 + x) ( m'' - c )
			= a + m + (1 + x) m''
	\end{equation}
	%%%%%%
	%
	
	%
	%%%
	$OPT$%%
	%$j = 1, \ldots, m'$%%%%%
	$J_{1, j}$%$J_{2, j}$%$j$%%%%%%%%%%%%%
	%%%%%%%%
	$J_{3, j}$%%%%%%%
	%
	%
	%%%%
	\begin{equation} \label{thm:VL.LBm.eq.3}
		V_{OPT}(\sigma_1) = 2 m' + (2 + x)m''
	\end{equation}
	%%%%%%
	%
	
	%
	%%%%$\sigma_2$%%%%%
	%
	%
	$J_{2,j} \hspace{1mm} (j = 1, \ldots, m')$%%%%%%%%%
	$m''$%%job $J'_{1,j} \hspace{1mm} (j = 1, \ldots, m'')$
	%%%%%%%
	%%%%$r(J'_{1,j}) = 1 - x$%$d(J'_{1,j}) = 1 + x$%%%%%%
	%%%%job%%%%$S'_1$%%%%
	%
	%%
	$m''$%%job $J'_{2,j} \hspace{1mm} (j = 1, \ldots, m'')$
	%%%%%%%
	%%%%$r(J'_{2,j}) = 1$%$d(J'_{2,j}) = 1 + 2x$%%%%%%
	%%%%job%%%%$S'_2$%%%%
	%
	%
	$ON$%
	$S'_1$%%%%%%1%%job%%%%%%$B'$%%%%%%
	$S'_2$%%%%%%1%%job%%%%%%$B$%%%%%%
	%%%
	$\tilde{b}$%%%%
	$B'$($B$)%$S'_1$($S'_2$)%job%%%%%%%
	$S_1$%$S_2$%%%%%%%%%%%%
	%%%1%%%%$2x$%%%%$ON$%%%%%%%%%%
	%
	%
	$ON$%
	$S'_1 \cup S'_2$%job%1%%%%%%$C$%%%%%%%
	$\tilde{c}_1$%%%%
	%
	%%%%%
	$ON$%%%%%%%$2x$%%%%%%%%%%%%%
	$S'_1 \cup S'_2$%job%%%%%%2%%%%%%$C$%%%%%%
	%%%
	$\tilde{c}_2$%%%%
	%
	%%%%%
	$ON$%%%%%%%%%$3x$%%%%%%%%%%%%%
	$ON$%
	$S'_1 \cup S'_2$%job%%%%%%1%%%%%%$A$%%%%%%%
	$\tilde{a}$%%%%
	%
	%%%%%
	$ON$%%%%%%%%%%$2+x$%%%%%%%%%%%%%
	%
	%
	%%%%%%%%(\ref{thm:VL.LBm.eq.1})%%%
	\begin{equation} \label{thm:VL.LBm.eq.4}
		\tilde{b} \leq b + b'
	\end{equation}
	\begin{equation} \label{thm:VL.LBm.eq.5}
		\tilde{c}_1 + \tilde{c}_2 \leq c = m - a - b - b'
	\end{equation}
	\begin{equation} \label{thm:VL.LBm.eq.6}
		\tilde{a} \leq 2m'' - \tilde{b} - \tilde{c}_1 - 2 \tilde{c}_2
	\end{equation}
	%%%%%%
	%
	%
	%%%%
	\begin{eqnarray} \label{thm:VL.LBm.eq.7}
		V_{ON}(\sigma_2) 
			&\leq& 2 (a - \tilde{a}) + b + b' + 2 x \tilde{b} + 2 x \tilde{c}_1 + 3x \tilde{c}_2 + (2 + x) \tilde{a} \nonumber \\
			&\leq& 2 m' + x (\tilde{b} + \tilde{c}_1 + \tilde{c}_2) + 2xm''  \mbox{\hspace{6mm} (%~(\ref{thm:VL.LBm.eq.1.1})%(\ref{thm:VL.LBm.eq.1.2})%(\ref{thm:VL.LBm.eq.6})%%)}  \nonumber \\
			&\leq& 2 m' + x (m - a) + 2xm'' \mbox{\hspace{6mm} (%~(\ref{thm:VL.LBm.eq.4})%(\ref{thm:VL.LBm.eq.5})%%)} 
	\end{eqnarray}
	%%%%%%
	%%%\com{%%%%%%%%%%%%%}
	%
	
	%
	%%%
	$OPT$%%
	%$j = 1, \ldots, m''$%%%%%
	$J_{1, j}$%$J'_{2, j}$%%%%%%%%%%%%%
	%%%
	$J'_{1, j}$%$J_{2, j}$%%%%%%%%%%%%%
	%
	%
	%%%%
	\begin{equation} \label{thm:VL.LBm.eq.8}
		V_{OPT}(\sigma_2) = m + 4x m''
	\end{equation}
	%%%%%%
	%
	
	%
	%%%%
	$m$%%%%%%%$m' = m'' = m/2$%%%%%
	%(\ref{thm:VL.LBm.eq.2.1})%(\ref{thm:VL.LBm.eq.2.2})%(\ref{thm:VL.LBm.eq.3})%%%
	\begin{equation} \label{thm:VL.LBm.eq.even1}
		\frac{ V_{OPT}(\sigma_1) }{ V_{ON}(\sigma_1) } 
			\geq \frac{ 2 m' + (2 + x)m'' }{ a + m + (1 + x) m'' }
			= \frac{ 2 m + xm/2 }{ a + 3m/2 + xm/2 }
	\end{equation}
	%%(\ref{thm:VL.LBm.eq.7})%(\ref{thm:VL.LBm.eq.8})%%%
	\begin{equation} \label{thm:VL.LBm.eq.even2}
		\frac{ V_{OPT}(\sigma_2) }{ V_{ON}(\sigma_2) } 
			\geq \frac{ m + 4xm'' }{ 2 m' + x (m - a) + 2xm'' }
			= \frac{ m + 2xm }{ m + (2m - a)x }
	\end{equation}
	%%%%%%
	%
	$m$%%%%%%%
	$m' = (m + 1)/2$%$m'' = (m - 1)/2$%%%%%%%%
	%(\ref{thm:VL.LBm.eq.2.1})%(\ref{thm:VL.LBm.eq.2.2})%(\ref{thm:VL.LBm.eq.3})%%%
	\begin{equation} \label{thm:VL.LBm.eq.odd1}
		\frac{ V_{OPT}(\sigma_1) }{ V_{ON}(\sigma_1) } 
			\geq \frac{ 2 m + x(m - 1)/2 }{ a + (3m - 1)/2 + x(m - 1)/2 }
	\end{equation}
	%%(\ref{thm:VL.LBm.eq.7})%(\ref{thm:VL.LBm.eq.8})%%%
	\begin{equation} \label{thm:VL.LBm.eq.odd2}
		\frac{ V_{OPT}(\sigma_2) }{ V_{ON}(\sigma_2) } 
			\geq \frac{ m + 2x(m - 1) }{ m + 1 + (2m - a - 1)x }
	\end{equation}
	%%%%%%
	%
	%
	%%%%
	\begin{equation} \label{thm:VL.LBm.eq.9}
		\frac{ V_{OPT}(\sigma) }{ V_{ON}(\sigma) } 
			\geq \max_{x} \min_{a} \max 
				\left \{ 
					\frac{ V_{OPT}(\sigma_1) }{ V_{ON}(\sigma_1) }, 
					\frac{ V_{OPT}(\sigma_2) }{ V_{ON}(\sigma_2) } 
				\right \}
	\end{equation}
	%%%%%%%%
	%$m$%%%%%
	%(\ref{thm:VL.LBm.eq.even1}),(\ref{thm:VL.LBm.eq.even2}),(\ref{thm:VL.LBm.eq.odd1}),(\ref{thm:VL.LBm.eq.odd2}),(\ref{thm:VL.LBm.eq.9})%%%%
	%~\ref{tab:thm:VL.LBm}%%%%%%%%
	%
	
	\fi
	\ifnum \count11 > 0
	%
	%%%\com{%%%%%}
	%
	Consider an online algorithm $ON$. 
	First, 
	$\lceil \frac{m}{2} \rceil (= m')$ jobs $J_{1,j} \hspace{1mm} (j = 1, \ldots, m')$ such that $r(J_{1,j}) = 0$ and $d(J_{1,j}) = 1$ are given. 
	Let $S_1$ denote the set of these jobs.
	Next, 
	$m'$ jobs $J_{2,j} \hspace{1mm} (j = 1, \ldots, m')$ such that $r(J_{2,j}) = 1 + x$ and $d(J_{2,j}) = 2 + x$ arrive. 
	$x$ is set later. 
	Let $S_2$ denote the set of these jobs. 
	Let $A$ be the set of machines onto each of which $ON$ places at least one job from $S_1$ and $S_2$. 
	%%%\com{%%%%%%%%%%%%%%}
	%
	Let $B$ be the set of machines onto which $ON$ places at least one job from $S_1$ but $ON$ does not place a job from $S_2$. 
	Let $B'$ be the set of machines onto which $ON$ places at least one job from $S_2$ but $ON$ does not place a job from $S_1$. 
	Let $C$ be the set of machines onto which $ON$ places no jobs from either $S_1$ or $S_2$. 
	Let the number of machines in $A,B,B'$ and $C$ denote $a, b, b'$ and $c$, respectively. 
	Then, we have 
	\begin{equation} \label{thm:VL.LBm.eq.1}
		a + b + b' + c = m, 
	\end{equation}
	\begin{equation} \label{thm:VL.LBm.eq.1.1}
		a + b \leq m', 
	\end{equation}
	and 
	\begin{equation} \label{thm:VL.LBm.eq.1.2}
		a + b' \leq m'. 
	\end{equation}
	In the following, 
	we provide two inputs and first consider input $\sigma_1$. 
	$\lfloor \frac{m}{2} \rfloor (= m'')$ jobs $J_{3,j} \hspace{1mm} (j = 1, \ldots, m'')$ such that $r(J_{3,j}) = 0$ and $d(J_{3,j}) = 2 + x$ are issued. 
	Let $S_3$ be the set of these jobs. 
	Since $a \leq m'$ by definition, 
	$b + b' + c = m - a \geq m - m' = m''$ 
	by Eq.~(\ref{thm:VL.LBm.eq.1}). 
	Thus, 
	$ON$ can place each job in $S_3$ onto each machine in $B \cup B' \cup C$. 
	In this way, 
	$ON$ gains more profits than placing the jobs onto machines in $A$. 
	Then, 
	$ON$ places each of $\min \{ c, m'' \}$ jobs from $S_3$ onto each machine in $C$,  
	and places the remaining $\max \{ m'' - c, 0 \}$ jobs onto machines from either $B$ or $B'$. 
	Thus, 
	\begin{eqnarray*} %\label{thm:VL.LBm.eq.2}
		V_{ON}(\sigma_1) 
			&\leq& 2 a + b + b' + (2 + x) \min \{ c, m'' \} + (1 + x) \max \{ m'' - c, 0 \}\\
			&=& a + m - c + (2 + x) \min \{ c, m'' \} + (1 + x) \max \{ m'' - c, 0 \} \mbox{\hspace{6mm} (by Eq.~(\ref{thm:VL.LBm.eq.1}))}. 
	\end{eqnarray*}
	When $c \geq m''$, 
	\begin{equation} \label{thm:VL.LBm.eq.2.1}
		V_{ON}(\sigma_1) 
			\leq a + m - c + (2 + x) m'' + (1 + x) 0 
			\leq a + m - m'' + (2 + x) m'' 
			= a + m + (1 + x) m'', 
	\end{equation}
	where the second inequality follows from $c \geq m''$. 
	When $c < m''$, 
	\begin{equation} \label{thm:VL.LBm.eq.2.2}
		V_{ON}(\sigma_1) 
			\leq a + m - c + (2 + x) c + (1 + x) ( m'' - c )
			= a + m + (1 + x) m''. 
	\end{equation}
	On the other hand, 
	for each $j = 1, \ldots, m'$, 
	$OPT$ places both $J_{1, j}$ and $J_{2, j}$ onto the $j$th machine, 
	and places each $J_{3, j}$ onto each of the remaining machines. 
	Hence, 
	\begin{equation} \label{thm:VL.LBm.eq.3}
		V_{OPT}(\sigma_1) = 2 m' + (2 + x)m''. 
	\end{equation}
	Second, 
	consider the input $\sigma_2$. 
	After $J_{2,j} \hspace{1mm} (j = 1, \ldots, m')$ are given, 
	$m''$ jobs $J'_{1,j} \hspace{1mm} (j = 1, \ldots, m'')$ arrive 
	such that $r(J'_{1,j}) = 1 - x$ and $d(J'_{1,j}) = 1 + x$. 
	Let $S'_1$ denote the set of these jobs. 
	Then, 
	$m''$ jobs $J'_{2,j} \hspace{1mm} (j = 1, \ldots, m'')$  arrive 
	such that 
	$r(J'_{2,j}) = 1$ and $d(J'_{2,j}) = 1 + 2x$. 
	Let $S'_2$ denote the set of these jobs. 
	$\tilde{b}$ denotes the number of machines from $B'$ onto which $ON$ places at least one job from $S'_1$ 
	plus the number of machines from $B$ onto which $ON$ places at least one job from $S'_2$. 
	If $ON$ places at least one job from $S'_1$ ($S'_2$)
	onto a machine from $B'$ ($B$), 
	then $ON$ gains the profit of $2x$ per machine
	in addition to the profits of jobs in $S_1$ and $S_2$. 
	Let $\tilde{c}_1$ denote the number of machines from $C$ 
	each of which $ON$ places one job from $S'_1 \cup S'_2$ onto. 
	Then, 
	$ON$ gains the profit of $2x$ per machine. 
	Let $\tilde{c}_2$ denote the number of machines from $C$ 
	each of which $ON$ places at least two jobs from $S'_1 \cup S'_2$ onto. 
	Then, 
	$ON$ gains the profit of at most $3x$ per machine. 
	Let $\tilde{a}$ denote the number of machines from $A$ 
	each of which $ON$ places at least one job from $S'_1 \cup S'_2$ onto. 
	Then, 
	$ON$ gains the profit of $2+x$ per machine. 
	By the above definitions and Eq.~(\ref{thm:VL.LBm.eq.1}), 
	\begin{equation} \label{thm:VL.LBm.eq.4}
		\tilde{b} \leq b + b', 
	\end{equation}
	\begin{equation} \label{thm:VL.LBm.eq.5}
		\tilde{c}_1 + \tilde{c}_2 \leq c = m - a - b - b', 
	\end{equation}
	and 
	\begin{equation} \label{thm:VL.LBm.eq.6}
		\tilde{a} \leq 2m'' - \tilde{b} - \tilde{c}_1 - 2 \tilde{c}_2. 
	\end{equation}
	Thus, we have 
	\begin{eqnarray} \label{thm:VL.LBm.eq.7}
		V_{ON}(\sigma_2) 
			&\leq& 2 (a - \tilde{a}) + b + b' + 2 x \tilde{b} + 2 x \tilde{c}_1 + 3x \tilde{c}_2 + (2 + x) \tilde{a}  \nonumber \\
			&\leq& 2 m' + x (\tilde{b} + \tilde{c}_1 + \tilde{c}_2) + 2xm'' \mbox{\hspace{6mm} (by Eqs.~(\ref{thm:VL.LBm.eq.1.1}), (\ref{thm:VL.LBm.eq.1.2}), and (\ref{thm:VL.LBm.eq.6}))}  \nonumber \\
			&\leq& 2 m' + x (m - a) + 2xm''.  \mbox{\hspace{6mm} (by Eqs.~(\ref{thm:VL.LBm.eq.4}) and (\ref{thm:VL.LBm.eq.5}))} 
	\end{eqnarray}
	On the other hand, 
	for every $j = 1, \ldots, m''$, 
	$OPT$ places both $J_{1, j}$ and $J'_{2, j}$ onto one machine.
	Additionally, 
	$OPT$ places both $J'_{1, j}$ and $J_{2, j}$ onto the remaining machines. 
	%%%\com{%%ed:other%remainning%}
	%
	%
	Hence, 
	\begin{equation} \label{thm:VL.LBm.eq.8}
		V_{OPT}(\sigma_2) = m + 4x m''. 
	\end{equation}
	If $m$ is even, 
	then $m' = m'' = m/2$. 
	By Eqs.~(\ref{thm:VL.LBm.eq.2.1}), (\ref{thm:VL.LBm.eq.2.2}), and (\ref{thm:VL.LBm.eq.3}), 
	\begin{equation} \label{thm:VL.LBm.eq.even1}
		\frac{ V_{OPT}(\sigma_1) }{ V_{ON}(\sigma_1) } 
			\geq \frac{ 2 m' + (2 + x)m'' }{ a + m + (1 + x) m'' }
			= \frac{ 2 m + xm/2 }{ a + 3m/2 + xm/2 }. 
	\end{equation}
	By Eqs.~(\ref{thm:VL.LBm.eq.7}) and (\ref{thm:VL.LBm.eq.8}), 
	\begin{equation} \label{thm:VL.LBm.eq.even2}
		\frac{ V_{OPT}(\sigma_2) }{ V_{ON}(\sigma_2) } 
			\geq \frac{ m + 4xm'' }{ 2 m' + x (m - a) + 2xm'' }
			= \frac{ m + 2xm }{ m + (2m - a)x }. 
	\end{equation}
	If $m$ is odd, 
	$m' = (m + 1)/2$ and $m'' = (m - 1)/2$. 
	By Eqs.~(\ref{thm:VL.LBm.eq.2.1}), (\ref{thm:VL.LBm.eq.2.2}), and (\ref{thm:VL.LBm.eq.3}), 
	\begin{equation} \label{thm:VL.LBm.eq.odd1}
		\frac{ V_{OPT}(\sigma_1) }{ V_{ON}(\sigma_1) } 
			\geq \frac{ 2 m + x(m - 1)/2 }{ a + (3m - 1)/2 + x(m - 1)/2 }, 
	\end{equation}
	and by Eqs.~(\ref{thm:VL.LBm.eq.7}) and (\ref{thm:VL.LBm.eq.8}), 
	\begin{equation} \label{thm:VL.LBm.eq.odd2}
		\frac{ V_{OPT}(\sigma_2) }{ V_{ON}(\sigma_2) } 
			\geq \frac{ m + 2x(m - 1) }{ m + 1 + (2m - a - 1)x }. 
	\end{equation}
	Therefore, 
	\begin{equation} \label{thm:VL.LBm.eq.9}
		\frac{ V_{OPT}(\sigma) }{ V_{ON}(\sigma) } 
			\geq \max_{x} \min_{a} \max 
				\left \{ 
					\frac{ V_{OPT}(\sigma_1) }{ V_{ON}(\sigma_1) }, 
					\frac{ V_{OPT}(\sigma_2) }{ V_{ON}(\sigma_2) } 
				\right \}
	\end{equation}
	and 
	we can have the lower bounds in Table~\ref{tab:thm:VL.LBm} 
	using Eqs.~(\ref{thm:VL.LBm.eq.even1}),(\ref{thm:VL.LBm.eq.even2}),(\ref{thm:VL.LBm.eq.odd1}),(\ref{thm:VL.LBm.eq.odd2}) and (\ref{thm:VL.LBm.eq.9}) 
	for each $m$. 
	\fi
\end{proof}
\begin{table*}
\begin{center}
	\renewcommand{\arraystretch}{1.5}
	%\caption{%$m(\geq 3)$%%%}
	\caption{Lower bounds for each $m(\geq 3)$.}
	\begin{tabular}{|l|l|l|l|}
		\hline
			 $m$ & Lower Bound & $a$ & $x$ \\
		\hline
			3 & $7/6 \geq 1.166$ & 0 & 3 \\ %1.166666666
		\hline
			4 & $\frac{22-2\sqrt{2}}{17} \geq 1.127$ & 0 & $2+2\sqrt{2}$ \\ %1.1277395809
		\hline
			5 & $\frac{420-15\sqrt{7}}{333} \geq 1.142$ & 1 & $\frac{-1 + \sqrt{7}}{2}$\\ %1.14208327428
		\hline
			6 & $\frac{51-6\sqrt{2}}{41} \geq 1.140$ & 2 & $\sqrt{2}$ \\ %1.14042339788
		\hline
			7 & $\frac{280-70\sqrt{11}}{227} \geq 1.158$ & 1 & $\frac{1+\sqrt{11}}{2}$ \\ %1.1589837885
		\hline
	\end{tabular}
	\begin{tabular}{|l|l|l|l|}
		\hline
			 $m$ & Lower Bound & $a$ & $x$ \\
		\hline
			8 & $28/25 \geq 1.12$ & 2 & $2/3$ \\
		\hline
			9 & $9/8 \geq 1.125$ & 2 & 1 \\
		\hline
			10 & $\frac{290-15\sqrt{2}}{239} \geq 1.124$ & 3 & $\sqrt{2}$ \\ %1.12463094797
		\hline
			11 & $\frac{704-11\sqrt{22}}{582} \geq 1.120$ & 2 & $\frac{1+\sqrt{22}}{3}$ \\ %1.12097152344
		\hline
			$\infty$ & $\frac{48 - 2 \sqrt{2}}{41} \geq 1.101$ & $m/4$ & $\sqrt{2}$ \\ %1.10174567988
		\hline
	\end{tabular}
	\label{tab:thm:VL.LBm}
\end{center}
\end{table*}
\ifnum \count10 > 0
%
%%%\com{%%%%%%}
%
%

%
%

%
\fi
\ifnum \count11 > 0
%
%%%\com{%%%%%}
%

%
\fi
\section{Conclusions} \label{Concl}
\ifnum \count10 > 0
%
%%%\com{%%%%%%}
%
%
%%%%%
%%%%%%%%%%%%%%%%%%%%
%%%%%%%%%%%%%%%%%%%%%%
%
%%%%
$m = 2$%%%%$3$%%%%%%%
%%%%%%%%%%%%%%%%%%%%$4/3$%$3$%%%%%%%%%%
%
%%%
%$m$%%%%%
%%%%%%%%%%%%%%%%%%%%%%%%%%%%%
%
%
%%%%%%%%%%%%%%%%%%%%%%%%%%%%
%
%
(i) 
%%%%%%%%%
preemption%%%%%%%%%%%%
preemption%%%%%%%%%%
general profit case%%%%%%%%%%%%%%%%%%%%%%%%%%%%%%%%
(ii)
%%%%%%%%%%%%%%%%%%%%%%%%%%%
%
%
(iii)
%%%%%%%%%%
uniform profit case%%%%%%%%%%%%%%
%%%%%%%%%%%%%%%%%%%%%%%%%%%%%%%%%%
%
%

%
\fi
\ifnum \count11 > 0
%
%%%\com{%%%%%}
%
In this paper, 
we have proposed a novel variant of the interval scheduling problem focusing on best-effort services. 
For this variant, 
we have proved that the competitive ratios of an online greedy algorithm are at most $4/3$ and $3$ for $m = 2$ and $m \geq 3$, respectively. 
Also, 
we have shown a lower bound on the competitive ratio of any deterministic algorithm for each $m$. 
We finish the paper by providing some open questions: 
 (i) 
In the setting studied in the paper, 
preemption is not allowed. 
Then, 
if preemption is allowed, 
can we design a competitive algorithm for the general profit case? 
(ii)
Will randomization help to improve our results?
(iii)
An obvious open problem is to close the gaps between our lower and upper bounds for the uniform profit case.
In addition, 
we should discuss offline algorithms for our variant. 
\fi
%

%\section*{Acknowledgments}
%
%This work was supported by JSPS KAKENHI Grant Number 26730008 and Cyber Physical System Integrated IT Platform project. 

\newpage
\appendix

\section{Assignment Examples} \label{ap.sec:assex}
\ifnum \count10 > 0
%
%%%\com{%%%%%%}
%
%%%%%%%
$c$-interval%%%%%%%%
$ge$-interval%%%%%%%%
$oe$-interval%%%%%%%%
$n$-interval%%%%%%%%%%%%%%
%
%
%%%%%%%%
$GR$%%%%%%%%%%%%%%%%%%
%%%%%$OPT$%%%%%$oe$-%%%%%
$GR$%%%%%$p$-%%%%%%%%%%
%%%%%%%%%%%%%%%%
$GR$%$oe$-%%%%%
$GR$%$p$-%%%%%%%%%
%
%
%%%%
%~\ref{fig:fig_assex1}%%%%%%
%%%%%%%%%
$GR$%$J_{1}$%%%$[1, 2]$
$OPT$%$J_{4}$%%%$[0, 1]$
%%%%%%%
%
%%%%
%%%%
$GR$%$J_{1}$%%%$[1, 2]$
$GR$%$J_{4}$%%%$[0, 1]$
%%%%%%%%%
%
%
%%%
%%%%%%%%%%%%%%%%%%%%%%%%%%%%%%%%%%%
%
%%%%
\ref{fig:fig_assex1}%%%
%%%%%%%
$GR$%$J_{2}$%%%%%%%%%%%%%%%
\fi
\ifnum \count11 > 0
%
%%%\com{%%%%%}
%
In all the figures of this section, 
$c$-intervals ($ge$-intervals, $oe$-intervals and $n$-interval, respectively)
are shown in blue (green, red and yellow, respectively) squares. 
Also, 
only $GR$'s jobs and machines are shown.
Our assignments are realized as matchings 
between $oe$-intervals of $OPT$'s jobs and $p$-intervals of $GR$'s jobs. 
However, for ease of presentation, 
the assignments are presented as matchings 
between $oe$-intervals of ``$GR$'s jobs'' and $p$-intervals of $GR$'s jobs. 
For example, 
in the situation of Fig.~\ref{fig:fig_assex1}, 
the routine assigns 
the interval $[1, 2]$ of $GR$'s $J_{1}$ 
to 
the interval $[0, 1]$ of $OPT$'s $J_{4}$ in fact. 
However, 
the assignment is described as the one from 
the interval $[1, 2]$ of $GR$'s $J_{1}$ 
to 
the interval $[0, 1]$ of $GR$'s $J_{4}$ 
in the figure. 
In addition, 
we use situations which cannot happen to explain assignments. 
For example, in Fig.~\ref{fig:fig_assex1}, 
$GR$ cannot place $J_{2}$ onto the second machine according to its definition
but places it onto the first machine. 
\fi
\subsection{General $m$} \label{ap.sec:anym}
\ifnum \count10 > 0
%
%%%\com{%%%%%%}
%
%~\ref{fig:fig_assex1}%%%
%%%%%%%%
$OPT$%%%%%$J_{4}$%$oe$-interval $[0, 1]$%%
$GR$%%%%%%%%%%%%%%%%$J_{1}$%$c$-interval $[1, 2]$%%%%%%%%
%%%%%
%
%%%%
%%%%$GR$%%%%$J_{5}$%2%%%%%%%%%%%%%
%%%%%%
$OPT$%%%%%$J_{4}$%$oe$-interval $[0, 1]$%%
$GR$%%%%%$J_{5}$%$c$-interval $[1, 2]$%%%%%%
$OPT$%%%%%$J_{5}$%$oe$-interval $[0, 1]$%%
$GR$%%%%%$J_{1}$%$c$-interval $[1, 2]$%
%%%%%
%%%%%
%
%%%%
$J_{4}$%%%%%%%%
$J_{5}$%
$GR$%%%%%$J_{5}$%$c$-interval $[1, 2]$%%%%%%%%%%%%%%
\fi
\ifnum \count11 > 0
%
%%%\com{%%%%%}
%
In Fig.~\ref{fig:fig_assex1}, 
at first, 
the routine assigns 
the $c$-interval $[1, 2]$ of $GR$'s $J_{1}$,
which is placed on the first machine, 
to 
the $oe$-interval $[0, 1]$ of $OPT$'s $J_{4}$ 
in the left figure. 
However, 
after $GR$ places $J_{5}$ onto the second machine, 
the routine reassigns 
the $c$-interval $[1, 2]$ of $GR$'s $J_{5}$ 
to 
the $oe$-interval $[0, 1]$ of $OPT$'s $J_{4}$, 
and assigns 
the $c$-interval $[1, 2]$ of $GR$'s $J_{1}$ 
to 
the $oe$-interval $[0, 1]$ of $OPT$'s $J_{5}$ 
in the right figure. 
Of course, it is possible that 
the routine does not change the assignment of $J_{4}$ and 
assigns 
the $c$-interval $[1, 2]$ of $J_{5}$ 
to 
the $oe$-interval $[0, 1]$ of $J_{5}$.
\fi
\ifnum \count12 > 0
\begin{figure*}[ht]
	 \begin{center}
	  \includegraphics[width=130mm]{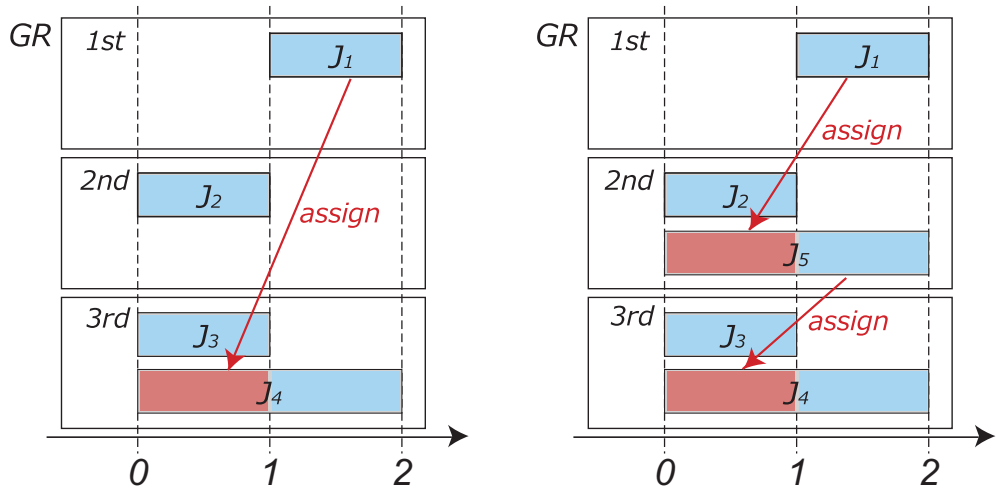}
	 \end{center}
	 \caption{
\ifnum \count10 > 0
%
%%%\com{%%%%%%}
%
%
%
\fi
\ifnum \count11 > 0
%
%%%\com{%%%%%}
%
Assignment example 1 for general $m$.
\fi
			}
	\label{fig:fig_assex1}
\end{figure*}
\fi
\ifnum \count10 > 0
%
%%%\com{%%%%%%}
%
%~\ref{fig:fig_assex2}%%%
$J_{4}$%$oe$-interval $[0, 1]$%%
$J_{1}$%$c$-interval $[0, 1]$%%%%%%%
$J_{4}$%$oe$-interval $[1, 2]$%%%
$J_{4}$%$n$-interval $[2, 3]$%%%%%%%%
$J_{2}$%$c$-interval $[3, 4]$%%%%%%%
\fi
\ifnum \count11 > 0
%
%%%\com{%%%%%}
%
In Fig.~\ref{fig:fig_assex2}, 
the routine assigns 
the $c$-interval $[0, 1]$ of $J_{1}$ 
to 
the $oe$-interval $[0, 1]$ of $J_{4}$. 
Since the $n$-interval $[2, 3]$ of $J_{4}$ exists, 
the routine assigns 
the $c$-interval $[3, 4]$ of $J_{2}$ 
to 
the $oe$-interval $[1, 2]$ of $J_{4}$. 
\fi
\ifnum \count12 > 0
\begin{figure*}[ht]
	 \begin{center}
	  \includegraphics[width=130mm]{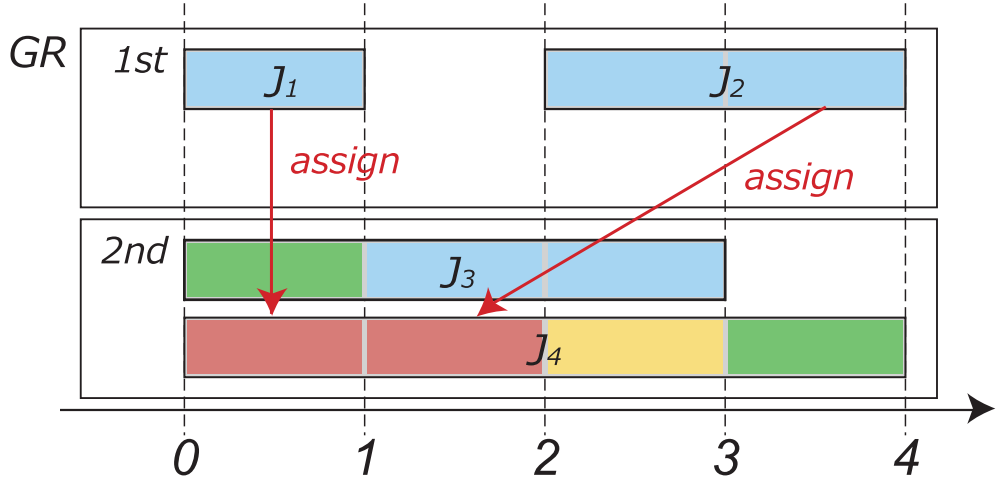}
	 \end{center}
	 \caption{
\ifnum \count10 > 0
%
%%%%\com{%%%%%%}
%
%
\fi
\ifnum \count11 > 0
%
%%%%\com{%%%%%}
%
Assignment example 2 for general $m$.
\fi
			}
	\label{fig:fig_assex2}
\end{figure*}
\fi
\ifnum \count10 > 0
%
%%%\com{%%%%%%}
%
%~\ref{fig:fig_assex3}%%%
%%%%%%Case~2.1%%%%%
$J_{3}$%$oe$-interval $[0, 1]$%%
$J_{1}$%$c$-interval $[1, 2]$%%%%%%%
$J_{4}$%$oe$-interval $[1, 2]$%%%
$J_{3}$%$c$-interval $[1, 2]$%%%%%%%
%
%%%%%%Case~2.2%%%%%
$J_{5}$%$oe$-interval $[1, 2]$%%
$J_{6}$%$c$-interval $[2, 3]$%%%%%%%
\fi
\ifnum \count11 > 0
%
%%%\com{%%%%%}
%
In Fig.~\ref{fig:fig_assex3}, 
the routine executes Case~2.1 assigning 
the $c$-interval $[1, 2]$ of $J_{1}$ 
to 
the $oe$-interval $[0, 1]$ of $J_{3}$, 
and 
assigning 
the $c$-interval $[1, 2]$ of $J_{3}$ 
to 
the $oe$-interval $[1, 2]$ of $J_{4}$. 
Then, 
the routine executes Case~2.2 and 
assigns 
the $c$-interval $[2, 3]$ of $J_{6}$ 
to 
the $oe$-interval $[1, 2]$ of $J_{5}$. 
\fi
\ifnum \count12 > 0
\begin{figure*}[ht]
	 \begin{center}
	  \includegraphics[width=130mm]{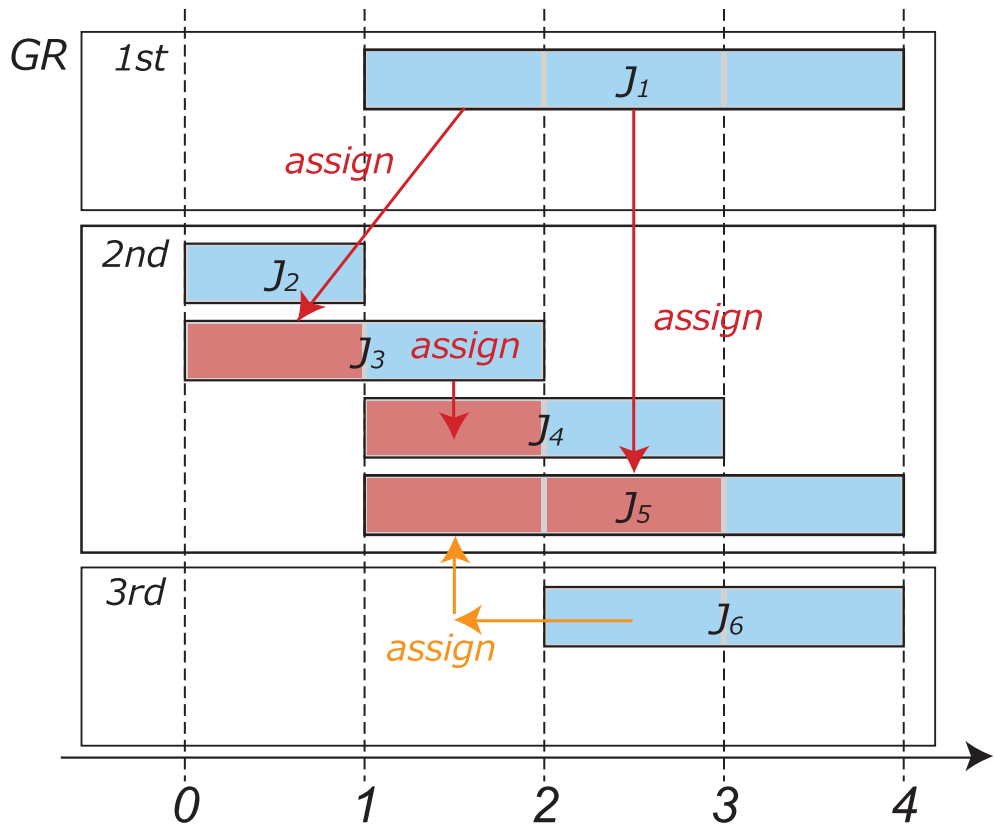}
	 \end{center}
	 \caption{
\ifnum \count10 > 0
%
%%%%\com{%%%%%%}
%
%
%
\fi
\ifnum \count11 > 0
%
%%%%\com{%%%%%}
%
Assignment example 3 for general $m$.
\fi
			}
	\label{fig:fig_assex3}
\end{figure*}
\fi
\subsection{$m=2$} \label{ap.sec:m2}
\ifnum \count10 > 0
%
%%%\com{%%%%%%}
%
%%%%%%%%%
$p$-fraction $f'$%%$J_{2}$%$c$-interval $[1, 2]$%%%%%%
$p$-fraction $f_{1}$%%$J_{3}$%$c$-interval $[2, 3]$%%%%%%
$p$-fraction $f_{2}$%%$J_{4}$%$ge$-interval $[2, 3]$%%%%%%
%
%%%%%%Case~2.2%%%%%
$J_{4}$%$oe$-interval $[1, 2]$%%
$J_{2}$%$c$-interval $[1, 2]$%
$J_{3}$%$c$-interval $[2, 3]$%
$J_{4}$%$ge$-interval $[2, 3]$
%%%%%%%
%
%%%%
%%%%%%Case~2.1%%%%%
$J_{5}$%$oe$-interval $[2, 3]$%%
$J_{4}$%$ge$-interval $[2, 3]$%%%%%%%
\fi
\ifnum \count11 > 0
%
%%%\com{%%%%%}
%
In Fig.~\ref{fig:fig_assexm2}, 
the $c$-interval $[1, 2]$ of $J_{2}$ contains the $p$-fraction $f'$ called in the routine. 
Also, 
the $c$-interval $[2, 3]$ of $J_{3}$ 
and 
the $ge$-interval $[2, 3]$ of $J_{4}$ 
contain 
the $p$-fraction $f_{1}$ 
and 
the $p$-fraction $f_{2}$, 
respectively, called in the routine. 
The routine executes Case~2.2 and 
assigns 
the $c$-interval $[1, 2]$ of $J_{2}$, 
the $c$-interval $[2, 3]$ of $J_{3}$ and 
the $ge$-interval $[2, 3]$ of $J_{4}$ 
to 
the $oe$-interval $[1, 2]$ of $J_{4}$. 
Then, 
the routine executes Case~2.1 and 
assigns 
the $ge$-interval $[2, 3]$ of $J_{4}$  
to 
the $oe$-interval $[2, 3]$ of $J_{5}$. 
\fi
\ifnum \count12 > 0
\begin{figure*}[ht]
	 \begin{center}
	  \includegraphics[width=130mm]{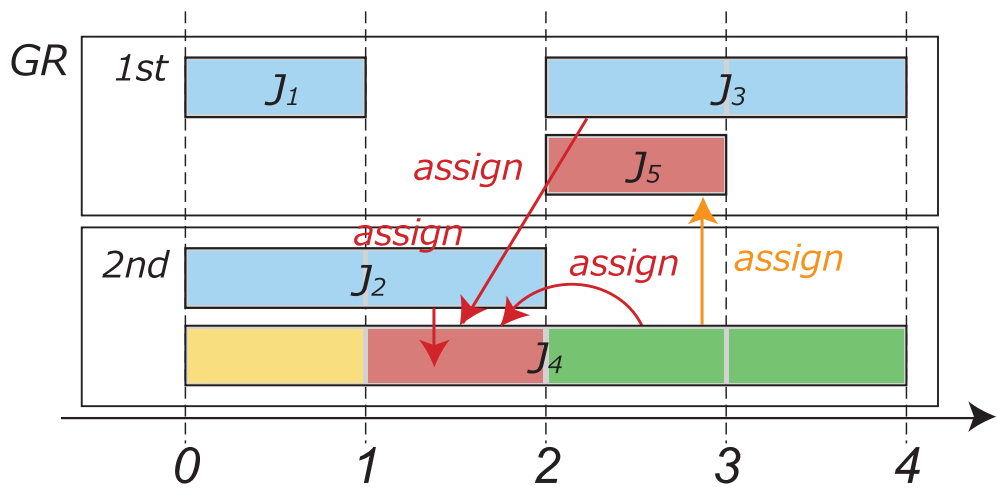}
	 \end{center}
	 \caption{
\ifnum \count10 > 0
%
%%%%\com{%%%%%%}
%
%
\fi
\ifnum \count11 > 0
%
%%%%\com{%%%%%}
%
Assignment example for $m = 2$.
\fi
			}
	\label{fig:fig_assexm2}
\end{figure*}
\fi

\end{document}